\documentclass[11pt]{amsart}
\usepackage[a4paper]{geometry}
\usepackage{mathtools}
\usepackage{bm}
\usepackage{amsmath}
\usepackage{amsthm}
\usepackage{amssymb}
\usepackage{graphicx}
\usepackage[multidot]{grffile}
\usepackage[unicode=true,pdfusetitle,
 bookmarks=true,bookmarksnumbered=false,bookmarksopen=false,
 breaklinks=false,pdfborder={0 0 1},backref=section,colorlinks=true]
 {hyperref}

\DeclareMathOperator{\diag}{diag} 
\DeclareMathOperator{\Ps}{S} 

\providecommand{\lemmaname}{Lemma}
\providecommand{\theoremname}{Theorem}
\providecommand{\propositionname}{Proposition}

\newtheorem{thm}{\protect\theoremname}
\newtheorem{prop}[thm]{\protect\propositionname}
\newtheorem{lem}[thm]{\protect\lemmaname}

\newcommand{\R}{\mathbb{R}}
\newcommand{\RP}{\mathbb{RP}}

\newcommand{\Z}{\mathbb{Z}}

\newcommand{\T}{\mathbb{T}}

\renewcommand{\epsilon}{\varepsilon}

\usepackage{amsaddr}

\begin{document}
\title{Monodromy in Prolate Spheroidal Harmonics}
\author{Sean R.~Dawson, Holger R.~Dullin, Diana M.H.~Nguyen}
\address[A1,A2,A3]{School of Mathematics and Statistics,\\The University of Sydney, Australia}
\email{holger.dullin@sydney.edu.au}
\begin{abstract}
We show that spheroidal wave functions viewed as the essential part of 
the joint eigenfunction of two commuting operators of $L_2(S^2)$ has 
a defect in the joint spectrum that makes a global labelling of the 
joint eigenfunctions by  quantum numbers impossible. To our knowledge this
is the first explicit demonstration that quantum monodromy exists in 
a class of classically known special functions. 
Using an analogue of the Laplace-Runge-Lenz vector 
we show that the corresponding classical Liouville integrable system 
is symplectically equivalent to the C.~Neumann system.
To prove the existence of this defect we construct a classical integrable system that is the 
semi-classical limit of the quantum integrable system of commuting operators.
We show that this is a semi-toric system with a non-degenerate focus-focus point,
such that there is monodromy in the classical and the quantum system.
\end{abstract}

\maketitle

\section{Introduction}

Prolate spheroidal wave functions are important and well known special functions 
that appear when separating variables in problems that have the symmetry of prolate ellipsoids.
Classical references on spheroidal wave functions are \cite{WhitWatson65,Meixner54,Flammer57,Stratton59,arscott64}.
One inspiration for our work is the general theory of separation of variables developed in \cite{Miller77,boyer76}.
There spheroidal harmonics appear as the joint eigenfunctions of two commuting operators
on the Hilbert space $L_2(S^2)$. 
The two operators are constructed from separation of variables 
in spheroidal coordinates. In the spherical limit the spheroidal harmonics reduce to the well 
known spherical harmonics. In this paper we study the joint spectrum of these two commuting 
operators and show that the lattice of joint eigenvalues has a global defect.
Even though prolate spheroidal wave functions are very well studied special functions 
this observation about the joint spectrum seems to be new.

Another inspiration of our work is the study of  quantum and Hamiltonian monodromy
in integrable systems, specifically so-called semi-toric integrable systems.
They have two degrees of freedom and one simple global integral which is an $S^1$ action.
The global study of Liouville integrable systems was initiated in \cite{Duistermaat80}.
In a subsequent paper  \cite{CushDuist88} it was shown that classically and quantum 
mechanically the spherical pendulum has Hamiltonian and quantum monodromy, respectively.
It was realised that classically \cite{Matveev96,Zung97} and quantum mechanically \cite{VuNgoc99}
monodromy is caused by a so-called focus-focus equilibrium point of the classical system.
More recently a global classification of semi-toric integrable systems has been achieved
\cite{VuNgoc09} and in the present work we discuss an interesting example of a semi-toric system.
Currently the global classification of Liouville integrable systems does not allow for general 
types of singularities and hence the general theory of Kalnins and Miller on separation of variables
does not provide examples for the current theory of semi-toric systems, unless they have some rotational symmetry, 
and hence a global $S^1$ action; hence separation in spheroidal coordinates is singled out.

In two recent papers \cite{DW18} and \cite{CDEW19} we have used separation in spheroidal 
coordinates for the Kepler problem in space and the  harmonic oscillator in space, respectively, and shown that
both problems -- when considered in prolate spheroidal variables -- have Hamiltonian and quantum monodromy.
The present paper grew out of the realisation that an even simpler problem, namely the free particle, 
can be studied in a similar vain, and leads to similar results, namely monodromy in the joint spectrum.
As in the two previous works it is crucial for this approach that the system under consideration is 
superintegrable. In the Kepler problem and the harmonic oscillator superintegrability implies that the 
flow of the Hamiltonian is periodic with constant period, and hence it is possible to consider symplectic 
reduction with respect to this flow, viewed as an action of the group $S^1$.
The reduced system inherits two constants of motion which are the separation constants from 
the separation of variables. In the present example of the free particle the orbits of the 
Hamiltonian are not periodic orbits, but instead straight lines. Thus we need to consider reduction 
not with respect to a compact group $S^1$ but with respect to the non-compact group $\R^1$.
Even though there are no general theorems about reduction in this case it turns out that the reduction
can be performed nicely and elegantly using the invariants of the Hamiltonian flow. 
This leads to the classical analogue of the commuting operators described by 
Kalnins and Miller \cite{boyer76}, and we then show using singular reduction with respect to 
the global $S^1$ action (the angular momentum about the $z$-axis) that the system is 
semi-toric and has a non-degenerate focus-focus point and hence monodromy.

The third inspiration for our work is to connect the two threads described above:
separation of variables including the corresponding special functions on the one hand and the global 
theory of integrable systems on the other hand. 
Special functions related to (confluent) Fuchsian equation beyond the (confluent) hypergeometric equation 
are for example discussed in \cite{arscott64,slayvanov00}. The spheroidal wave equation is 
a particular case of the confluent Heun equation, see \cite{ronveaux95} and the references therein.
In our setting spheroidal harmonics are joint eigenfunctions of two commuting operators, and 
we show that a defect in the joint spectrum of these operators can be understood from the analysis of the corresponding Liouville integrable system. This is more than a WKB analysis of the solutions, 
but instead takes into account global information about the action variables of the 
integrable system. Nevertheless, we remark that the essence of the defect could have been observed 
by analysing well known asymptotic expansions \cite{AS} for the eigenvalues of the spheroidal wave 
equation; but to our knowledge such an analysis has not been presented before.

The solutions of the Helmholtz equation inside the prolate ellipsoid (aka the quantum
billiard in the prolate ellipsoid) has been studied in  \cite{WD02}, and monodromy 
was found in the joint spectrum. Since this is a system with three commuting operators
(as opposed to two in the current paper)
the radial equation has to be included and this leads to two coupled boundary 
value problems that were numerically solved in \cite{WD02}. In the present problem 
we only study the angular wave equation and find monodromy also in this simpler setting.

The plan of the paper is as follows. In section 2 we describe a reduction of the free particle
in $\R^3$ that leads to a reduced system with a Lie-Poisson structure of the algebra $e^*(3)$ of the 
Eulidean group of translations and rotations $E(3)$. To obtain an integrable system on 
the reduced space separation of variables in prolate spheroidal coordinates is employed 
in the next section. The centrepiece of the paper is the description of monodromy 
in the corresponding quantum system, which is obtained from separation of variables
of the Helmholtz equation in $\R^3$. We show that the joint spectrum of the two 
commuting operators has quantum monodromy. In particular this can be seen from 
the analysis of the classical asymptotic series for the eigenvalues in two distinct limits.
Then we show that the spheroidal harmonics integrable system is in fact symplectically 
equivalent to the integrable C.~Neumann system of a particle constrained to move
on a sphere with an added harmonic potential, which in this case has rotational symmetry.
The analysis of monodromy using well known asymptotic formulas is somewhat heuristic, 
and to prove monodromy we show that the underlying classically integrable spheroidal harmonics system 
(and hence the rotationally symmetric Neumann systems) is semi-toric with a non-degenerate 
focus-focus point corresponding to a doubly pinched torus.
%

%
%
%
%

\section{The Free Particle}

The free particle in $\R^3$ lives on the phase space $T^{*}\R^{3}\cong\R^{6}$
with global coordinates \linebreak{}
${\bm{Q}\coloneqq\left(x,y,z\right)^{T}}$ 
and $\bm{P}\coloneqq\left(p_{x},p_{y},p_{z}\right)^{T}$.
The Hamiltonian is simply $H=\frac{1}{2}\left(p_{x}^{2}+p_{y}^{2}+p_{z}^{2}\right)$
and the equations of motion are $\dot{\bm{Q}}=\bm{P}$ and $\dot{\bm{P}}=\bm{0}$.
The trajectories or geodesics are
\[
   \bm{Q}  =\bm{P}t+\bm{Q}_{0} ,  \quad  \bm{P}  =\bm{P}_{0}
\]
where $\bm{Q}_{0},\bm{P}_{0}$ are the intial position and momentum vectors, respectively,
and $t$ is time. In position space, the geodesics
are oriented lines through $\bm{Q}_{0}$ in the direction of $\bm{P}=\bm{P}_{0}$.
We can perform a symplectic reduction that identifies the oriented straight
lines of the flow of $H$ to points and so lowers the dimensionality
of the phase space from $6$ to $4$. 
We will see that this reduction also produces a compact configuration space,
which is the space of oriented lines through the origin, which is a sphere.
The conserved quantities are
the linear momenta $\bm{P}=\left(p_{x},p_{y},p_{z}\right)$ and the
angular momenta $\bm{L}\coloneqq\bm{Q}\times\bm{P}=\left(l_{x},l_{y},l_{z}\right)^T$,
since 
\[
\bm{L}\left(t\right)  =\left(\bm{Q}_{0}+t\bm{P}_{0}\right)\times\bm{P}_{0}\\
  =\bm{Q}_{0}\times\bm{P}_{0}\\
  =\text{constant}.
\]
By construction we have $\bm{P} \cdot \bm{L} = \bm{0}$.

The six invariants $\bm{P}$, $\bm{L}$ are closed under the standard Poisson bracket in $T^*\R^3$.
For example $\{ p_x, l_y \} = p_z$, and  $\{ l_x, l_y \} = l_z$, etc.
Assembling all such identities into a $6\times 6$ matrix $B$ gives 
\footnote{
For a vector $\bm{v}\in\R^{3}$ the corresponding antisymmetric
hat matrix $\hat{\bm{v}}$ is defined by
\[
   \hat{\bm{v}}\bm{u}=\bm{v}\times\bm{u} \quad  \forall\bm{u}\in\R^{3}.
\]
Later we also use hat to denote the quantum operator corresponding to a classical observable; 
from the context it should be clear which one is meant.
}
\begin{equation} 
B=-\begin{pmatrix}\bm{0} & \hat{\bm{P}}\\
\hat{\bm{P}} & \hat{\bm{L}}
\end{pmatrix}\,.
\label{eq:Bpoi}
\end{equation}
The matrix $B$ is the matrix of a Lie-Poisson structure on $\R^6$ with 
coordinates $\bm{P}$ and $\bm{L}$. This Lie-Poisson structure is the algebra 
$e^*(3)$ corresponding to the Euclidean group $E(3)$, the group of 
isometries of Euclidean space $\R^3$.
In particular the components of $\bm{P}$ are generators of translations,
while the components of $\bm{L}$ are generators of rotations. 
Given a Hamiltonian $G$ the time evolution of any function $f(\bm{P}, \bm{L})$ 
is given by $\dot f = \{ f, G\} = \nabla f^t B \nabla G$ and thus
\begin{equation}
\dot{\bm{P}} = -\bm{P} \times \nabla_L G, \quad
\dot{\bm{L}} =  -\bm{P} \times \nabla_P G   -\bm{L} \times \nabla_L G \,.
\end{equation}
The Poisson structure $B$ has rank $4$ with two Casimirs $C_1=\bm{P}\cdot\bm{P}=2E$
and $C_2=\bm{P}\cdot\bm{L}=0$, such that  $ B \nabla C_i = 0$.
The first Casimir $C_1=2E$ is often set to $1$ by normalisation 
of the speed of the particle,
whereas the second Casimir $C_2$
is an identity that states that $\bm{L}$ is orthogonal to $\bm{P}$.
In addition to these 6 basic invariants an analogue of the Laplac-Runge-Lenz (LRL) vector
can be defined and we will discuss this in more detail in section \ref{LRL}.

Fixing the two Casimirs defines the reduced phase space of $T^{*}S^{2}$ as a subset of $\R^6$
with coordinates $\bm{P}$ and $\bm{L}$.
Here the sphere is defined in momentum space, and  reflects the constancy of 
the kinetic energy of the particle, while the tangent space to the sphere is the set of 
planes with normal vectors $\bm{P}$ in $\bm{L}$ space, hence $C_2 = 0$.
Every point on $T^*S^2$ represents a line (geodesic) in the original $T^*\R^3$
with direction $\bm{P}$ (the point on the sphere) and angular momentum 
$\bm{L}$ (the vector in the tangent space of the sphere). 
Note that $\bm{L}$ is a normal vector to the plane that contains
the geodesic and the origin, and the length of $\bm{L}$ is the distance of the 
geodesic to the origin divided by the value of $C_1$. 
There are four oriented lines with direction $\pm\bm{P}$ 
in a given plane with normal vector $\pm\bm{L}$.
Changing the orientation of the geodesic amounts to changing the sign of $\bm{P}$ and $\bm{L}$.
Changing the sign of $\bm{L}$ but not of $\bm{P}$ represents a parallel line with the same orientation
in the same plane that is passing on the other side of the origin.
Lastly, changing the sign of $\bm{P}$ but not of $\bm{L}$ represents a parallel line with the opposite orientation
that is passing on the other side of the origin.
Later we will identify any two such geodesics, which will lead to $T^*\RP^2$
instead of $T^*S^2$.

Since we have reduced by the dynamics of $H$ there are no dynamics defined 
on $T^*S^2$ at the moment. In the next section we are  going to define an integrable system 
on $T^*S^2$ by separating the free particle in spheroidal coordinates. 
The separation constant and the angular momentum will induce an 
integrable system on $T^*S^2$. 


\section{The Spheroidal Harmonics Integrable System}
\label{sec:SHIS}


Prolate ellipsoids are formed by rotating an ellipse around its focal
axis. Let the foci of the resulting ellipsoid be located at $\left(0,0,\pm a\right)$.
Prolate spheroidal coordinates are then defined by 
\begin{equation}
\begin{aligned}x & =a\sqrt{\left(\xi^{2}-1\right)\left(1-\eta^{2}\right)}\cos(\phi),\\
y & =a\sqrt{\left(\xi^{2}-1\right)\left(1-\eta^{2}\right)}\sin(\phi),\\
z & =a\xi\eta \,,
\end{aligned}
\label{eq:To prolate xi eta}
\end{equation}
where $\eta\in\left[-1,1\right]$, $\xi\in\left[1,\infty\right)$
and $\phi\in\left[0,2\pi\right) = S^1$. Each point of $\R^{3}$
is associated with the intersection of the ellipsoid described by
(\ref{eq:To prolate xi eta}), a confocal hyperboloid and a plane.
These surfaces correspond to fixed $\xi$, $\eta$ and $\phi$ respectively.
The Hamiltonian of the free particle in prolate spheroidal coordinates
is 
\begin{equation}
H=\frac{1}{2a^2}\left( \frac{ (1-\eta^2)p_{\eta}^{2} + (\xi^{2}-1)p_\xi^2}{(\xi^2 - \eta^2)} + 
    \frac{p_{\phi}^{2}}{\left(1- \eta^{2}\right)\left(\xi^{2}-1\right)}\right)\label{eq:Ham prolate}
\end{equation}
where  $p_{\eta}$, $p_{\xi}$ and $p_{\phi}$ are the momenta conjugate to $\eta$, $\xi$, $\phi$, respectively. Clearly $p_\phi$ is a constant angular momentum, since $H$ is independent of $\phi$.
To separate the variables observe that
\[
    0 = ( H - E) 2 a^2 (\xi^2 - \eta^2) = G(\eta, p_\eta) - G(\xi, p_\xi), 
\]
where
\begin{equation} \label{eqn:Gsep}
    G(q, p) =  ( 1 - q^2) (p^2 - 2 a^2 E)  + \frac{p_\phi^2}{1 - q^2}
\end{equation}
such that $G(\xi, p_\xi) = g = G(\eta, p_\eta)$ where $g$ is the separation constant. 
Substituting $E = H$ into $G$ gives 
\[
    G = \frac{ p_\eta^2 - p_\xi^2 }{\xi^2 - \eta^2} (1-\eta^2)(\xi^2 - 1) +
    p_\phi^2 \frac{\xi^2 - \eta^2}{(\xi^2 - 1)(1-\eta^2) }  \,.
\]
To convert this to the original variables observe that
\[
   |\bm{L}|^2 = \frac{(\xi^2 -1)(1 - \eta^2)}{(\xi^2 - \eta^2)^2} ( p_\xi \eta - p_\eta \xi)^2 + p_\phi^2 \left(  \frac{1 + \xi^2 - \eta^2}{(\xi^2 - 1)(1 - \eta^2)}    \right)
\]
and 
\[
   a^2( p_x^2 + p_y^2) = \frac{(\xi^2 -1)(1 - \eta^2)}{(\xi^2 - \eta^2)^2} ( p_\xi \xi - p_\eta \eta)^2 + p_\phi^2\frac{1}{(\xi^2 - 1)(1 - \eta^2)}
\]
such that 
\begin{equation} \label{eqn:Gorg}
    G = | \bm{L}|^2 - a^2 ( p_x^2 + p_y^2) \,.
\end{equation}
This is a function on $T^*S^2$, as is $p_\phi = L_z$, and it is easy to check that they 
have vanishing Poisson bracket. This can be computed in the original variables
$(\bm{Q}, \bm{P})$ with respect to the canonical bracket on $T^*\R^3$,
or in the variables  $(\bm{P}, \bm{L})$ after reduction to $T^*S^2$ 
with respect to the induced bracket $B$. In both cases $\{ |\bm{L}|^2, L_z \} = 0$
and also $\{ p_x^2 + p_y^2, L_z \} = 0$ and hence $\{ G, L_z \} = 0$.
We write $L_z$ for the function that maps a point $(\bm P, \bm L)$ to the coordinate $l_z$.
Thus we arrive at the main classical object of this paper: 
\begin{thm}[Spheroidal harmonics integrable system] \label{thm:shis}
Consider $\R^6$ with coordinates $( \bm{P}, \bm{L} )$ and Lie-Poisson structure of $e^*(3)$ with 
Poisson tensor $B$ given by \eqref{eq:Bpoi} and Casimirs $\bm{P} \cdot \bm{P} = 2E$ and $\bm{P} \cdot \bm{L} = 0$. 
The functions $(L_z, G)=( l_z,  l_x^2+l_y^2+l_z^2 - a^2( p_x^2 +p_y^2))$ define a Liouville integrable system on $T^*S^2$.
\end{thm}

We call this integrable system the (prolate) \emph{spheroidal harmonics integrable system}, 
since it arrises from separation of variables in spheroidal coordinates.
It is the classical analogue of the compact part of the spheroidal wave equation, 
whose solutions are known as spheroidal harmonics.
\footnote{The term spheroidal harmonics is used in  different ways in the literature. The strict use of ``harmonics'' refers to solutions of the Laplace equation. When considering the Laplacian in $\R^3$ separated in spheroidal coordinates the eigenfunctions are products of associated Legendre functions, however, one of them is evaluated outside the usual range $|z| < 1$, and is thus sometimes referred to as a spheroidal harmonic \cite[14.3]{DLMF}. Our use of the term spheroidal harmonic is different and serves as a ``reminder of the kinship with the spherical harmonics" \cite[17.4]{Press88}.
}
%
In this work we are only interested in the prolate spheroidal harmonics.
Formally the oblate case can be found by flipping the sign of $a^2$,
and this system is also Liouville integrable.
However, the dynamics in the oblate case are quite different and in particular does not exhibit monodromy,
so we do not consider this case in the present work.
The values of $(L_z, G)$ will be denoted by $(m, g)$.
For any separating coordinate system of $\R^3$ a similar construction 
can be carried out and will lead to an integrable system on $T^*S^2$
corresponding to that separating coordinate system. 
It is an interesting research project to study these integrable Hamiltonian systems
alongside the corresponding special functions.
In this paper we restrict our attention to spheroidal coordinates,
because, as we will show, it leads to a semi-toric system that exhibits Hamiltonian monodromy. 
The related quantum system has quantum monodromy.
In other words, the eigenvalues of the spheroidal wave equation exhibit monodromy.
Before we describe the spheroidal wave equation and its quantum monodromy in the next section, 
here we are going to describe some aspects of the dynamics of the spheroidal harmonics integrable system.
A detailed analysis including the proof that it is a semi-toric system with Hamiltonian monodromy 
is postponed to a later section.

The vector field that is generated by the Hamiltonian $L_z$ is given by $ B \nabla L_z $
which gives
\begin{equation} \label{eq:S1flow}
\dot{\bm{P}} = -\bm{P} \times \bm{e}_z, \quad 
\dot{\bm{L}} = -\bm{L} \times \bm{e}_z \,.
\end{equation}
The solution is a rotation of the first two components of $\bm{P}$ and $\bm{L}$ 
by the same amount; the third components are unchanged.
Thus the point $\bm{P}$ on $S^2$ is rotated about the $p_z$-axis,
while $\bm{L}$ in the  tangent space is rotated in the same way.
The north- and the south-pole of $S^2$ are fixed by this rotation, 
but then $\bm{L} = (l_x, l_y, 0)^T$ is not fixed, unless it vanishes.
A vector $\bm{L} = (0, 0, l_z)^T$ that is in the tangent space of a point 
$\bm{P} = (\cos \phi, \sin\phi, 0)$ on the equator of the sphere is fixed by this 
rotation, but the corresponding $\bm{P}$ is not.
This shows that the only fixed points of this $S^1$ action are
$\bm{P} = (0,0,\pm1)^T$, $\bm{L} = (0,0,0)^T$.
They correspond to geodesics along the $z$-axis,
i.e., lines through the two foci of the ellipsoid of the spheroidal coordinates.

The vector field that is generated by the Hamiltonian $G$ is 
\begin{equation} \label{eqn:GVF}
\dot{\bm{P}} = -2\bm{P} \times \bm{L}, \quad 
\dot{\bm{L}} = a^2 \bm{P} \times ( \bm{P} - \bm{e}_z p_z)  = -a^2 p_z \bm{P} \times \bm{e}_z\,.
\end{equation}
Clearly $\bm{L} = \bm{0}$ and $\bm{P} = \bm{e}_z p_z$ is an equilibrium point.
Moreover, for $\bm{P} = (p_x, p_y, 0)^T$ and $\bm{L} = (0,0,l_z)^T$ 
we have $\bm{L} = const$, $p_z = 0 = const$ and $\dot p_x =-2 l_z p_y$ and $\dot p_y = 2 l_z p_x$,
a periodic solution along the equator with orientation depending on the sign of $l_z$.
For $l_z = 0$ the equator is a circle of non-isolated equilibrium points of the flow of $G$.

In the limiting case $a \to 0$ the integral $G$ becomes the angular momentum squared.
In this limit the equations of motion can be solved explicitly in terms of trigonometric functions.
Since $\dot{\bm{L}} = 0$  the equation for $\dot{\bm{P}}$ is that of 
a rotation about the fixed axis $\bm{L}$. The period of these rotation 
is given by $\sqrt{ |\bm{L}|}$.
If instead of $G=|\bm{L}|^2$ we consider $|\bm{L}|$ as a Hamiltonian then the period is $2\pi$, 
and hence $|\bm{L}|$ is an action variable.
To see this just integrate
$\dot{\bm{P}} =- \bm{P} \times \bm{L}  / |\bm{L}|$ for constant non-zero $\bm{L}$.
The solution is a rotation of $\bm{P}$ about the fixed normal vector in the 
direction of $\bm{L}$. The only problem with the flow of $|\bm{L}|$ is 
that the vector field is not defined when $\bm{L} = \bm{0}$ 
and hence the flow does not define a global $S^1$ action. 
When instead the flow of $G = |\bm{L}|^2$ (for $a=0$) 
is considered the vector field simply vanishes when $\bm{L} = \bm{0}$, 
and so the whole sphere $|\bm{P}|=2E$ is a sphere of fixed points.

The spheroidal harmonics integrable system has a number of discrete symmetries. 
We restrict our attention to discrete symmetries that are canonical transformations. 
\begin{prop}[Discrete symmetries]
The group of linear discrete canonical symmetries of the spherical harmonics integrable 
system is $\Z_2 \times \Z_2$. For $s_i = \pm 1$, $i=1,2,3$ define 
$S = \diag(s_1, s_2, s_3)$ and  
$\tilde S = \diag( s_2 s_3, -s_1 s_3, s_1 s_2)$ so that a linear map of $(\bm{P}, \bm{L})$ is given by 
$(S \bm{P}, \tilde S \bm{L})$. The non-trivial elements of $\Z_2 \times \Z_2$ are obtained from
$S_1 = \diag( +, +, -)$, $S_2 = \diag( -, -, - )$ and $S_3 = \diag( -,-,+ )$.
\end{prop}
\begin{proof}
Since $S^{-1} = S^t$ the map $\bm{Q} \mapsto S \bm{Q}$ extends to a symplectic map as $(\bm{Q}, \bm{P}) \mapsto ( S \bm{Q}, S \bm{P})$.
The induced sign flip on the angular momentum is  $\bm{L} \mapsto \tilde S \bm{L}$ where $\hat S = \diag( s_2 s_3, -s_1 s_3, s_1 s_2)$ is found by computing the cross product $\bm{Q} \times \bm{P}$.
The integral $G$ is invariant under all such sign flips, since it is quadratic in components of $\bm{P}$ and $\bm{L}$.
In addition $L_z = x p_y - y p_x$ should be invariant under the discrete symmetry which requires $s_1s_2 = +1$.
Thus the discrete symmetries of the spheroidal harmonics system are 
$S_1 = \diag(++-)$, $S_2=\diag(---)$ and $S_3 = \diag(--+)$ together with the corresponding 
induced map $\tilde S_i$ on $\bm{L}$. Together with the identity they form the group $\Z^2 \times \Z^2$.
\end{proof}

In prolate spheroidal coordinates \eqref{eq:To prolate xi eta} the symmetry operations are realised as follows.
Changing the sign of $\eta$ changes the sign of $z$ but leaves $x$ and $y$ unchanged, 
so that $\eta \mapsto -\eta$ corresponds to the symmetry $S_1$.
Adding $\pi$ to $\phi$ changes the signs of $x$ and $y$ while $z$ is unchanged,
so that $\phi \mapsto \phi + \pi$ corresponds to the symmetry $S_3$.
The composition of both gives $S_2$.

\section{Quantum monodromy in prolate spheroidal harmonics}

Separation of variables of the Laplace equation or the Helmholtz equation in $\R^3$
in spheroidal coordinates leads to spheroidal harmonics. 
The classical reference on spheroidal harmonics is \cite{Stratton59,Meixner54,Flammer57},
and a few more modern ones are \cite{Press88,falloon03,Volkmer03,DLMF,Zhao17}.
We would like to mention that prolate spheroidal wave functions have 
found applications as band-limited functions \cite{slepian83}, 
also see \cite{xiao01,boyd04} and the references therein.
Here we will derive the spheroidal wave equation in the traditional way from the Schr\"odinger equation 
of the free particle separated in spheroidal coordinates. 
This will allow us to connect to the spheroidal harmonics integrable system
by way of semi-classical quantisation, a connection we need later to prove 
the existence of quantum monodromy.

The stationary Schr\"odinger equation for the free particle is $-\tfrac12 \hbar^2 \Delta \Psi = E \Psi$,
or we can think of it as Helmholtz's wave equation $\Delta \Psi + k^2 \psi = 0$.
Writing the Laplacian $\Delta$ in spheroidal coordinates \eqref{eq:To prolate xi eta} gives
\begin{equation}
\frac{1}{\left(\xi^{2}-\eta^{2}\right)}\left(\frac{\partial}{\partial\xi}\left((\xi^{2}-1)\frac{\partial\Psi}{\partial\xi}\right)+\frac{\partial}{\partial\eta}\left((1-\eta^{2})\frac{\partial\Psi}{\partial\eta}\right)\right)+\frac{1}{\left(1-\eta^{2}\right)\left(\xi^{2}-1\right)} \, \frac{\partial^{2}\Psi}{\partial\phi^{2}}=-\frac{2Ea^2}{\hbar^{2}}\Psi \,.
\label{eq:schro}
\end{equation}
Separation into product form $\Psi(\eta,\xi,\phi)=\psi_{\eta}(\eta)\psi_{\xi}(\xi)\psi_{\phi}(\phi)$
yields the simple equation 
\begin{equation}
\frac{\partial^{2}\psi_{\phi}}{\partial\phi^{2}}+ m^2 \psi_{\phi}=0
\label{eq:phi sep}
\end{equation}
and the (prolate angular) spheroidal wave equation
\begin{equation}
\hat G \psi_\eta = g \psi_\eta, \quad 
\hat G = -\frac{d}{d\eta}\left( ( 1 - \eta^2) \frac{d}{d\eta}\right) + \frac{m^2}{1-\eta^2} - \gamma^2 (1 - \eta^2), \quad
\gamma^2 = \frac{2Ea^2}{\hbar^2}
\label{eq:sep}
\end{equation}
with separation constants $m$ and $g$. 
The third separated equation is found by replacing $\eta$ by $\xi$.
The difference is in the domain $\eta \in [-1, 1]$ while $\xi \ge 1$.

For general values of $\gamma$ the equation \eqref{eq:sep} is a singular Sturm-Liouville equation.
It can be transformed into an equation with periodic (but still singular) coefficients, see, e.g.~\cite{arscott64}.
Viewed as a polynomial differential equation in the complex plane it can be transformed into the confluent Heun 
equation \cite{ronveaux95}.
The general Heun equation is the second order ordinary differential equation  of 
Fuchsian type with four regular singular points.
Letting two of the regular singular points coalesce leads to an irregular singular point.
The result is the confluent Heun equation.
%

The quantum integrable system (QIS) on the reduced space 
consists of two self-adjoint operators $\hat L_z$ and $\hat G$
acting on functions on the sphere $S^2$.
The eigenvalues $g_l^m$ of $G$ are those values of $g$ in \eqref{eq:sep} for which 
the solution of the spheroidal wave equation for $\eta$ leads to a smooth 
function $\psi_\eta \psi_\phi$ on the sphere.
In our treatment we ignore the equation for $\xi$ because it has 
no analogue in the classical spheroidal harmonics system.

The solution $\psi_\phi$ to the angular equation is proportional to linear combinations of 
$e^{\pm i m \phi}$ and  $2\pi$-periodicity in $\phi$ implies $e^{\pm i m 2 \pi} = 1$, 
and hence $m$ must be an integer. 
This integer $m$ is the quantum number for the $z$-component of the angular momentum $l_z = m \hbar$.

When $g$ is an eigenvalue $g_l^m$ of the singular Sturm-Liouville problem \eqref{eq:sep} 
the corresponding eigenfunction bounded on $(-1,1)$ 
(prolate angular) spheroidal wave function of the first kind, 
which we denote by $\Ps_l^m(\gamma,\eta )$. 
\footnote{The notation for the angular spheroidal wave function varies, see \cite{AS} for a table comparing various common notatinos. Our notation loosely follows \cite{AS}, but we prefer to write $l$ instead of $n$ as in \cite{Stratton59,Morse53}, and we write the indices ${}^m_l$ as in the associated Legendre polynomials.}
In the limit $\gamma\to 0$ these solutions degenerate to the associated Legendre polynomials 
 of the first kind $P_l^m(\eta)$.
For $\gamma\not = 0$ the spheroidal wave functions can be written as a (generally infinite) series of associated Legendre polynomials
\begin{equation}
      \Ps_l^m\left(\gamma,\eta\right)={\sum_{k=0,1}^{\infty}}{}'\,\, d_{k}^{lm}(\gamma)P_{m+k}^{m}\left(\eta\right)
\label{eq:Spheroidal in Legendre}
\end{equation}
where $d_{k}^{lm}$ are the expansion coefficients
and the prime on the summation indicates to sum over odd $k$ if $l-m$ is odd and over even $k$ if $l-m$ is even. 
Expressions for the resulting three term recursion relation that determines $d_k^{lm}$ can be found, e.g., 
in \cite[21.7.3]{AS}.

The product of the eigenfunctions of \eqref{eq:phi sep} and \eqref{eq:sep} gives the spheroidal harmonics 
\begin{equation}
      Z_l^m(\gamma,\eta, \phi )
      \coloneqq 
          \frac{1}{\sqrt{2\pi}}\Ps_l^m(\gamma,\eta)e^{im\phi}  =
          \sum_k d_k^{lm} Y_{m+k}^m(\eta, \phi) 
\end{equation}
expressed as a series of spherical harmonics $Y_l^m$.
In the limit $\gamma \to 0$ we have $Z_l^m = Y_l^m$. 
Normalisation on the sphere requires
\[
    \int_{0}^{2\pi}\int_{0}^{\pi}Z_{l}^{m}\left(Z_{l}^{m}\right)^{*}\sin\theta d\theta d\phi=1.
\]
and the $d_k^{lm}$ are chosen such that this holds, see, e.g., \cite{AS}.

%
We now consider the joint spectrum of the QIS $(\hat L_z, \hat G)$.
For periodicity in $\phi$ we need to require that the eigenvalue of $\hat L_z$ is $\hbar$ times an integer $m$. 
The eigenvalues $g_l^m$ of $\hat G$ can in general only be computed numerically. 
The Mathematica \cite{Mathematica12} function $\mathtt{SpheroidalEigenvalue}$\texttt{{[}$l,m,\gamma${]}}
gives the spheroidal eigenvalue $g_l^m$ of \eqref{eq:sep}.
From general results in microlocal analysis we know that in the semiclassical limit $\hbar \to 0$ 
the joint spectrum $(\hbar m,  \hbar^2 g_l^m)$ is locally a lattice.
%
For a fixed spheroidal coordinate system, i.e., a fixed value of $a$ decreasing $\hbar$ makes
this local lattice finer and finer, see Fig.~\ref{fig:evhbar}.

\begin{figure}
	\begin{centering}
		\includegraphics[width=4.5cm]{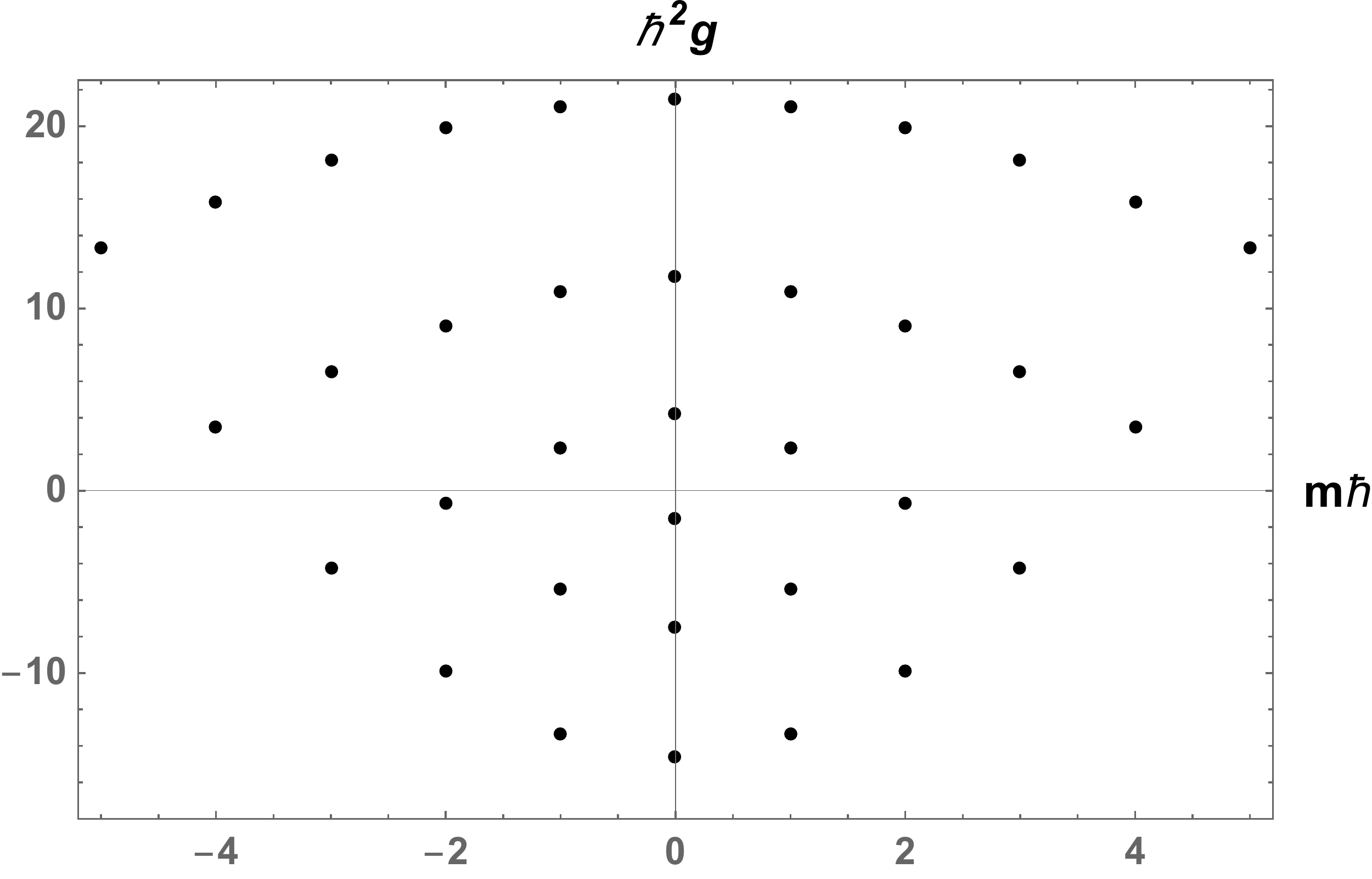}
		\includegraphics[width=4.5cm]{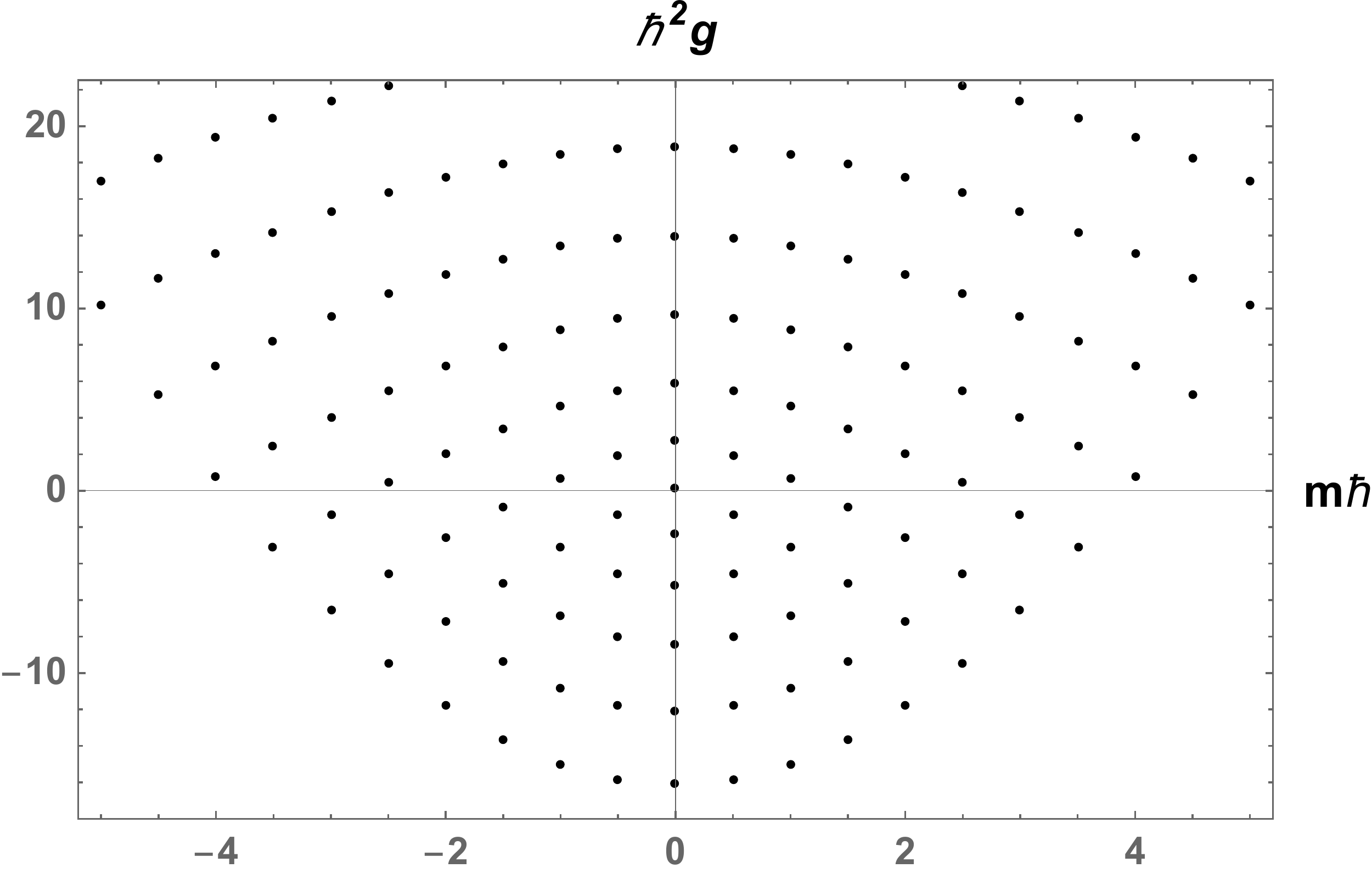}
		\includegraphics[width=4.5cm]{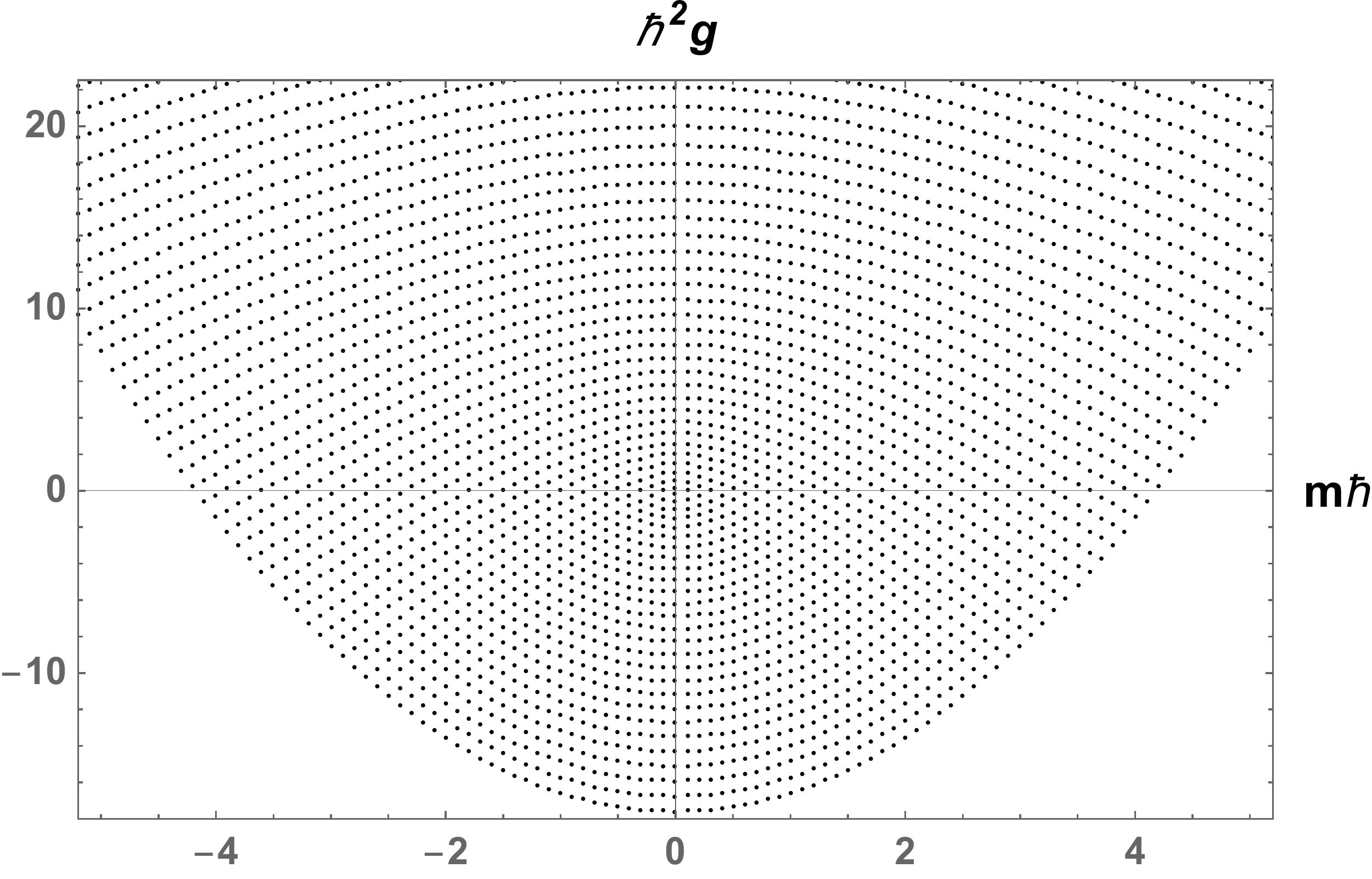}
	\par\end{centering}
	\caption{Joint spectrum $(\hbar m, \hbar^2 g_l^m)$ of the spheroidal harmonics with $2Ea^2 = 18$ for $\hbar = 1.0, 0.5, 0.1$ illustrating the semi-classical limit $\hbar \to 0$}
	\label{fig:evhbar}
\end{figure}

In the following we prefer to absorb $\hbar$ in the definition of the single parameter 
$\gamma = 2Ea^2/\hbar^2$ and present the scaled joint spectrum $(m, g_l^m)$.
When changing $\gamma$ the values and the distribution of the joint eigenvalues changes.
We are going to explain the structure of the joint spectrum and its dependence on $\gamma$ in the course of the paper.
Three examples of the joint spectrum are shown in Figure \ref{fig:ev3gamma} for $\gamma = 8,16,32$.
Note that this lattice is bounded below by a parabola 
(given by the critical values of the energy-momentum map, see below) 
but unbounded from above. 
We can observe that, locally the lattice
is isomorphic to $\mathbb{Z}^{2}$ thus allowing local assignments
of quantum numbers. However, there is a lattice defect at the origin,
and thus we do not have a global $\mathbb{Z}^{2}$ lattice, indicating
the presence of quantum monodromy. 

\begin{figure}
	\begin{centering}
		\includegraphics[width=6cm]{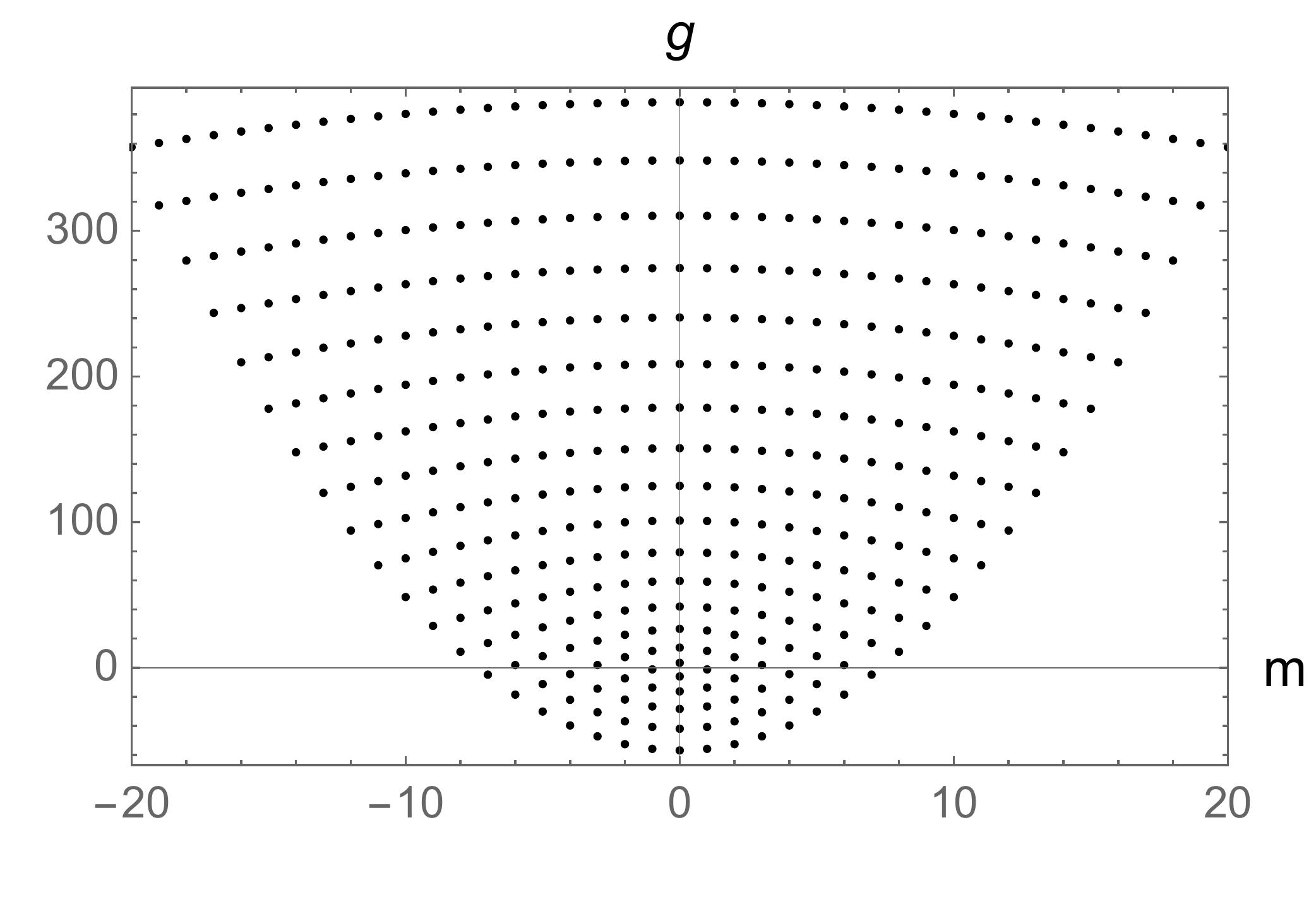}
		\includegraphics[width=6cm]{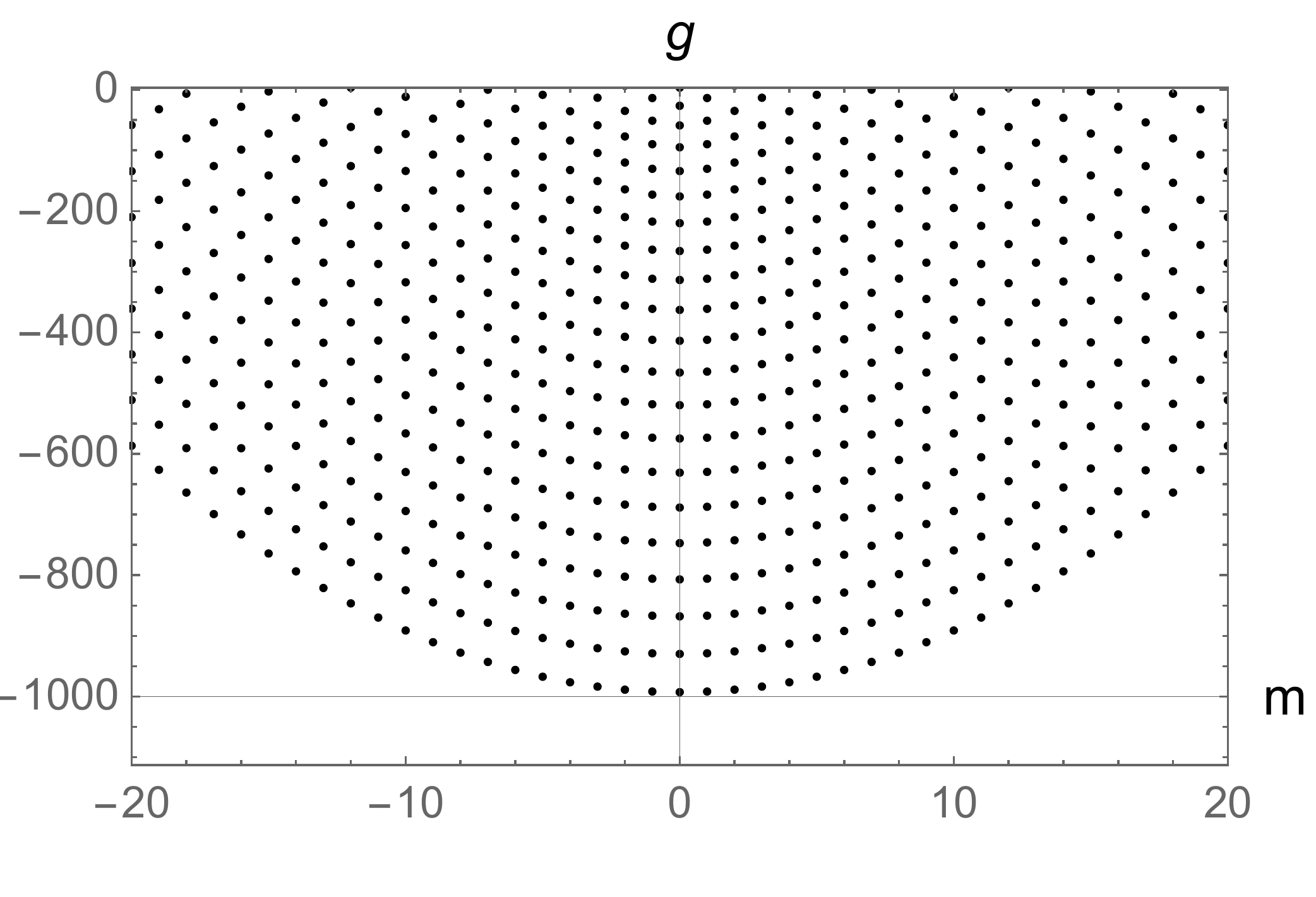}
	\par\end{centering}
	\caption{Joint spectrum $(m, g_l^m)$ of the spheroidal harmonics for $\gamma = 8, 32$.
	The asymptotic expansion for $g_l^m$ \eqref{eq:glmsmallgamma} is valid in the top part of the left figure,
	while \eqref{eq:glmlargegamma} is valid in the bottom part of the right figure.}
	\label{fig:ev3gamma}
\end{figure}


The spectrum of the spheroidal wave equation is well understood, and asymptotic 
expansions for the eigenvalues $g_l^m$ are well known \cite{Meixner54,AS,DLMF,arscott64}.
Here we are going to use these formulas to describes the quantum monodromy in the joint spectrum.

When $a \to 0$ the constant $\gamma \to 0$ and the operator $\hat G \to  | \bm{L}|^2$ 
becomes that of the associated Legendre equation with spectrum $g_l^m = l ( l +1)$
and corresponding eigenfunction the associated Legendre polynomial $P^m_l(\eta)$ for $-l \le m \le l$.
The spectrum is degenerate since $g_l^m$ is independent of $m$.
The labelling of eigenvalues in the spheroidal wave equation is continued from 
this limit for non-zero $a$. This means that in the Sturm-Liouville problem of the operator
$\hat G$ for given fixed integer $m$ the eigenvalue $g_l^m$ of the ground state is labelled by $l = |m|$.
The degeneracy is split for non-zero $\gamma$ and 
\begin{equation} \label{eq:glmsmallgamma}
g_l^m = l(l+1) - \frac12 \left( 1+  \frac{(2m-1)(2m+1)}{(2l-1)(2l+3) }\right) \gamma^2 + O(\gamma^4/l^2), 
\end{equation}
see, e.g., \cite{Meixner54,AS,DLMF,arscott64}. 
For fixed $\gamma$ this approximation is also good when $l$ is large and it can thus be understood 
as a semi-classical limit with fixed $a$ but large quantum number $l$ or correspondingly large values of the eigenvalue $g_l^m$.

When $\gamma \to \infty$ the spectrum also becomes simpler, but the limit is a bit more complicated. 
The leading order of the operator $\hat G$ is simply $-a^2 (\hat p_x^2 + \hat p_y^2)$.
The eigenvalues satisfy 
\begin{equation} \label{eq:glmlargegamma}
g_l^m = -\gamma^2 + ( 2(l-|m|) + 1) \gamma -\tfrac34 +m^2 - \tfrac12 (l-|m|)(l-|m|+1) + O(1/\gamma) ,
\end{equation}
see \cite{AS,DLMF,arscott64,Muller63}. 
Thus eigenvalues with the same value of $l - |m|$ and small $|m|$ are degenerate at leading order.
The limit of large $\gamma$ can be understood as the semiclassical limit where $\hbar \to 0$
for fixed value of $a$ for quantum numbers $l$ close to the ground state with $l = |m|$.

\begin{figure}
	\begin{centering}
		\includegraphics[width=10cm]{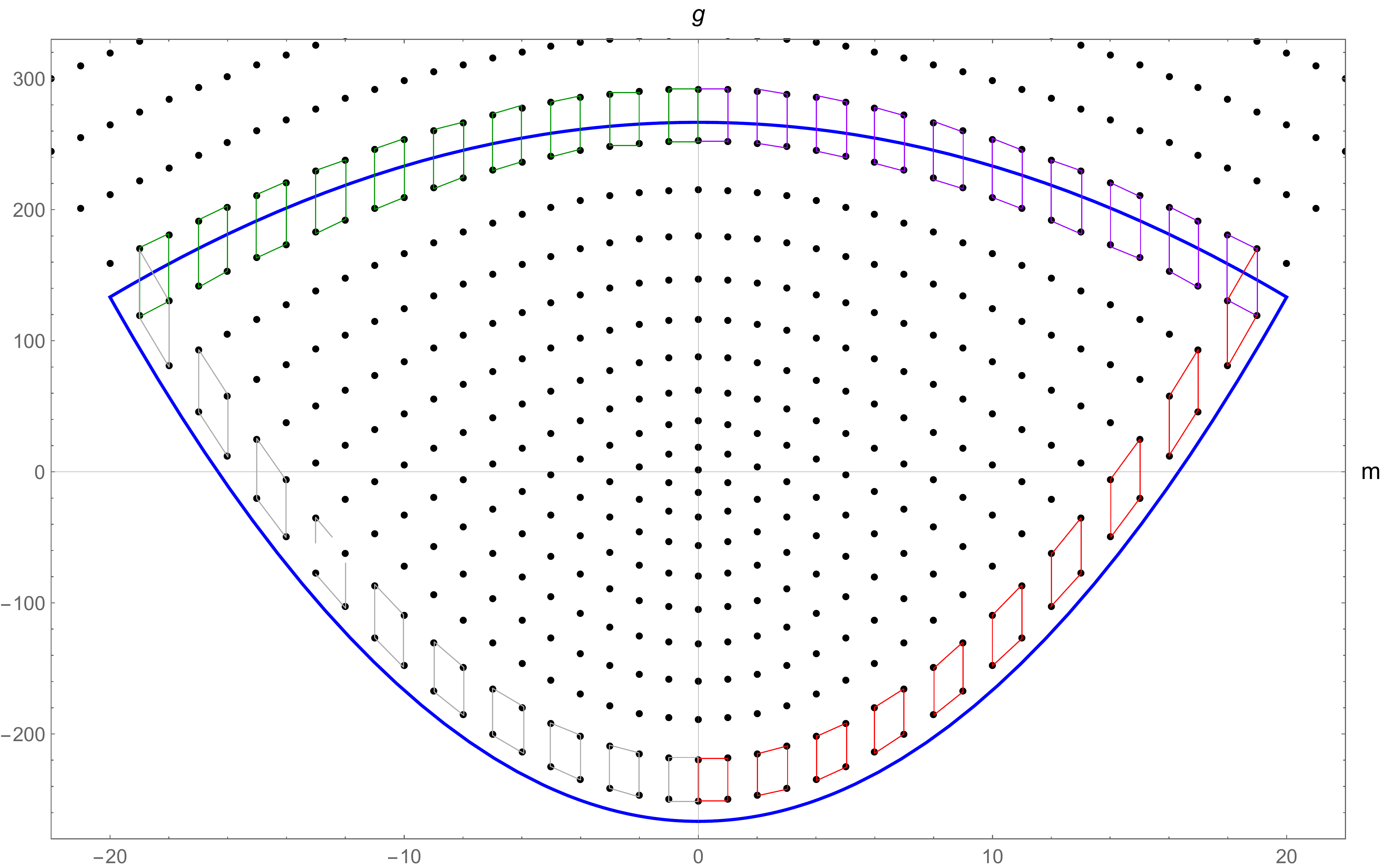}
		\par\end{centering}
	\caption{Joint spectrum $(m, g_l^m)$ of the spheroidal harmonics for $\gamma = 16$. 
	         A lattice unit cell is transported 
		around the origin. The  lower blue parabola is $g = -\gamma^2 + m^2$
		and the upper blue parabola is $2g = 2l^2 - \gamma^2 - m^2(\gamma/l)^2$ for $l = l^*$.  
		The red and purple cells are transports of $B_{l}^{m}$ and $T_{l}^{m}$
		respectively for positive $l^{*}$. The grey and green cells are those
		for negative $l^{*}$.\label{fig:Spheroidal-Eigenvalue}}
\end{figure}

Figure \ref{fig:Spheroidal-Eigenvalue} illustrates the monodromy about the origin. 
A unit cell is parallel transported along a path that encloses the origin.
As the basis vectors (say $v_{1}$ is the vertical vector and $v_{2}$ is the horizontal one) 
are fully transported around the loop, we observe that
$v_{1}$ stays constant whilst $v_{2}$ becomes $v_{2}+2v_{1}$. This
implies that we have a basis transformation according to 
\[
\begin{pmatrix}v_{1}'\\
v_{2}'
\end{pmatrix}=\begin{pmatrix}1 & 0\\
k & 1
\end{pmatrix}\begin{pmatrix}v_{1}\\
v_{2}
\end{pmatrix}
\]
where $k=2$. This integer is called the monodromy index.
In the figure the full loop is broken up into two symmetric half-loops,
each contributing half of the total monodromy. 
In the next section we will prove this by showing that in the classical 
phase space there are  isolated critical points of focus-focus type 
and the pre-image of the corresponding critical value is a doubly pinched torus.
Here we give a direct quantum mechanical interpretation of monodromy 
that is based on discrete symmetries and based on the well known 
asymptotic formulas \eqref{eq:glmsmallgamma} and \eqref{eq:glmlargegamma}

The monodromy along a loop in the joint spectrum around the origin can be analysed using the 
well known asymptotic formulas. Each formula is going to be evaluated along either the lower parabola where $l^2 - m^2 = 0$
in the joint spectrum or along a particular ``upper'' parabola where $l=l^*$ is constant but large.
For the former $g = -\gamma^2 + m^2$ and for the latter $g = l^2 - \gamma^2/2 - \tfrac12 m^2 (\gamma/l)^2$ for fixed $l=l^*$.
A unit cell in the joint spectrum is defined at $l=m=0$ and moved along the lower parabola.
Another unit cell in the joint spectrum is define at $l=l^*$, $m=0$
and transported along the upper parabola. The constant $\gamma$ is chosen such that 
the distance of the vertex of the two parabolas from the origin is the same, hence $\gamma = \sqrt{2/3} l^*$.
The two parabolas meet where $m = l^*$.
A unit cell near the bottom parabola is defined by its for corners as $B_l^m = (g_l^m, g_{l+1}^{m+1}, g_{l+2}^{m+1}, g_{l+1}^m)$
moving counterclockwise around the unit cell.
A unit cell near the top parabola is defined by its for corners as $T_l^m = (g_l^m, g_l^{m+1}, g_{l+1}^{m+1}, g_{l+1}^m )$ 
moving counterclockwise around the unit cell. The cell at the top has a natural labelling, which is inherited from the 
spherical harmonics limit. Now the cells are moved together to the point where the parabolas meet. 
There $B_{l^*-1}^{l^*-1}$ is compared with $T_{l^*}^{l^*-1}$. The 2nd and 3rd state in the two unit cells agree, and the 
last of $B$ with the first of $T$. 
Thus a basis transformation will add 1 unit to $l$. A mirror symmetric situation occurs 
for $l = -l^*$, and hence the total monodromy around the loop is 2.
The asymptotic formulas given above are stretched to their limits when trying to see
the equality of the three eigenvalues of the two unit cells. 
In particular the expansion at the bottom parabola given by \eqref{eq:glmlargegamma} is not very good when 
evaluated near $l=m=\sqrt{3/2}\gamma$. This is not surprising, 
and is not essential for our argument.
Note that sometimes eigenvalues $m$ are restricted to non-negative integers, which 
is obvious because of the symmetry of the spectrum under $m \to -m$. Even in joint spectra that
are hence somewhat arbitrarily cut in half the mismatch in the quantum numbers when comparing
the two basis cells is still present, even though a ``loop'' around the focus-focus point is not possible any more.

\begin{figure}
	\begin{centering}
		\includegraphics[width=4cm]{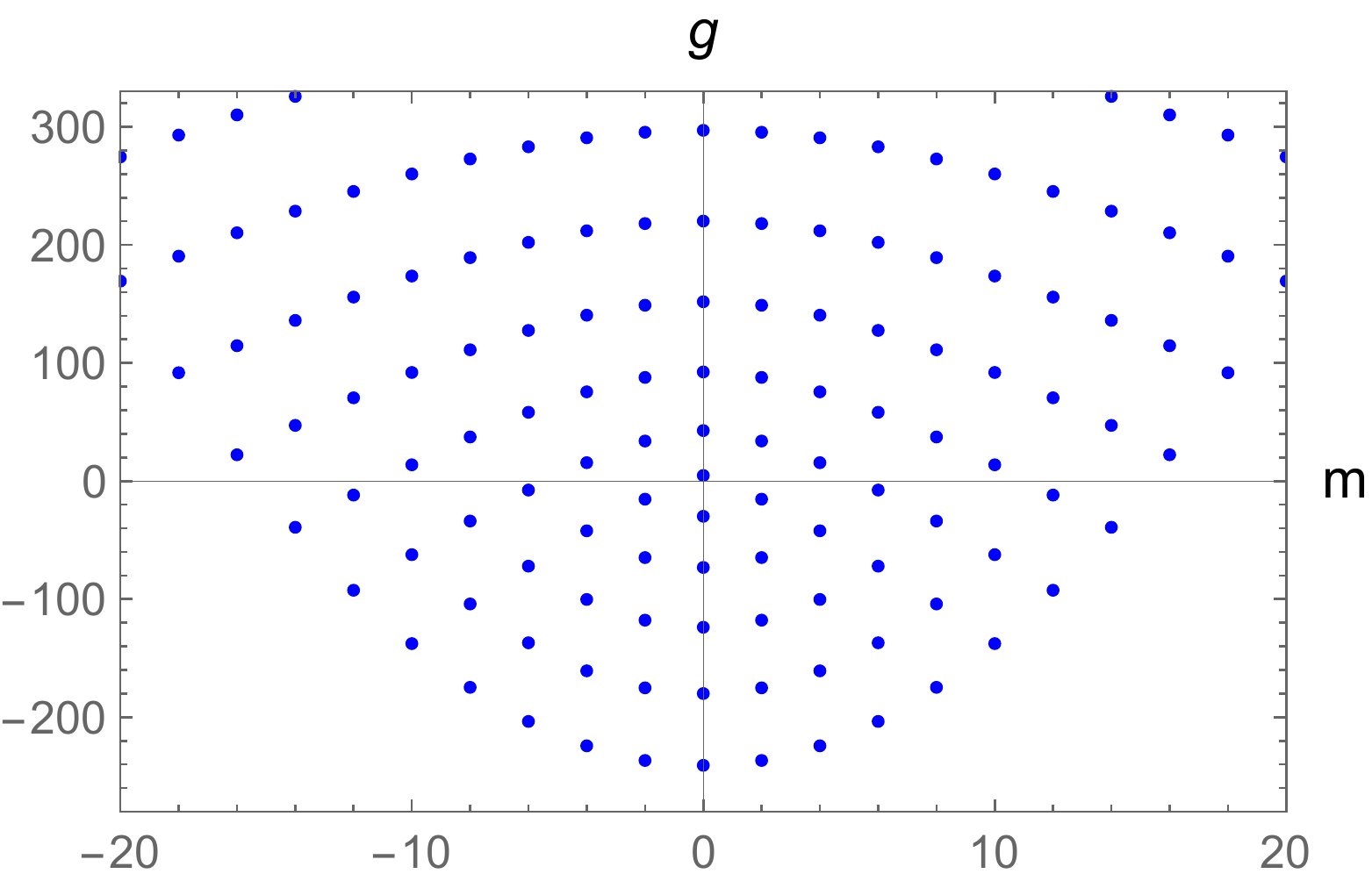}%
		\includegraphics[width=4cm]{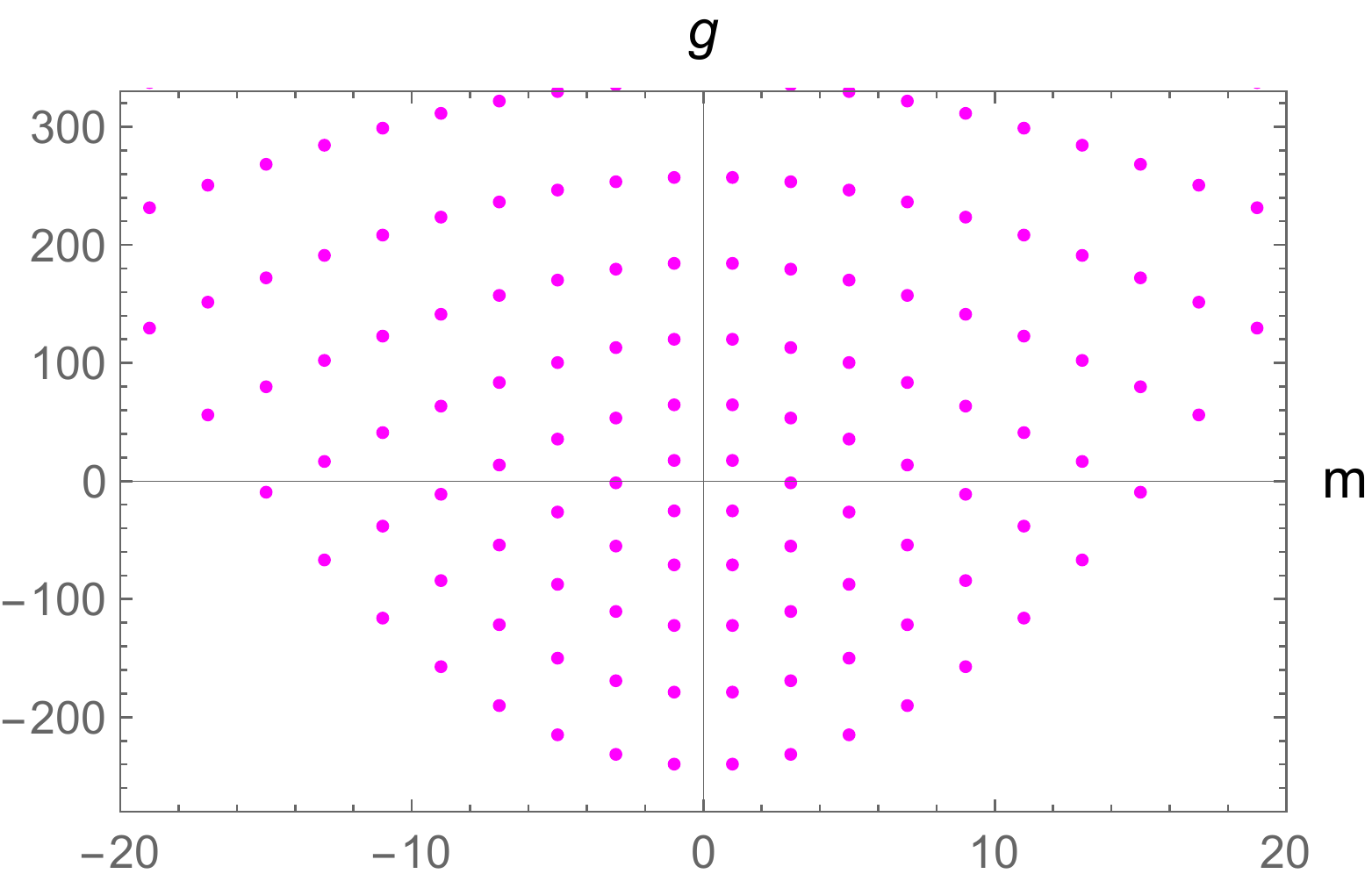}%
		\includegraphics[width=4cm]{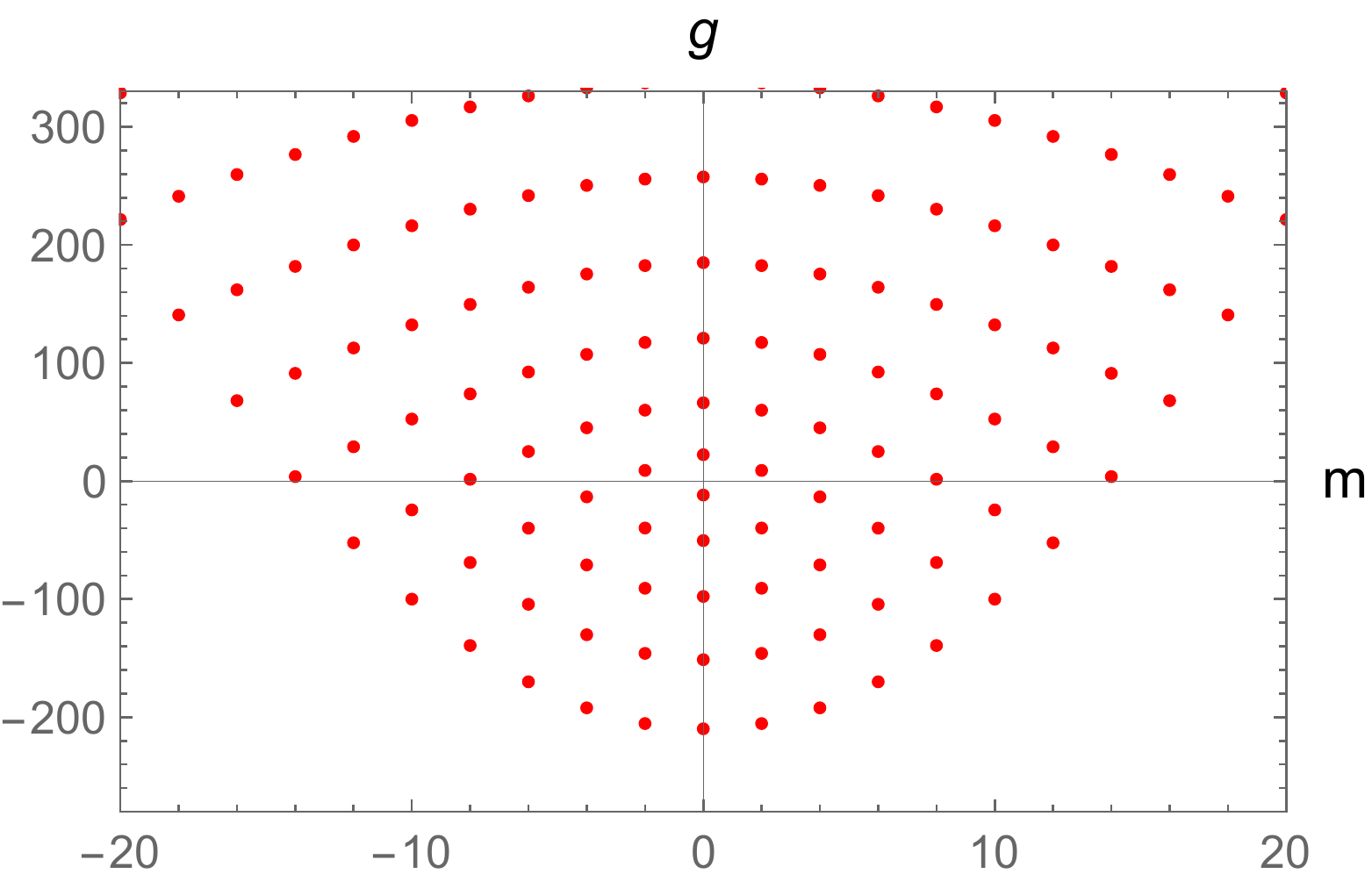}%
		\includegraphics[width=4cm]{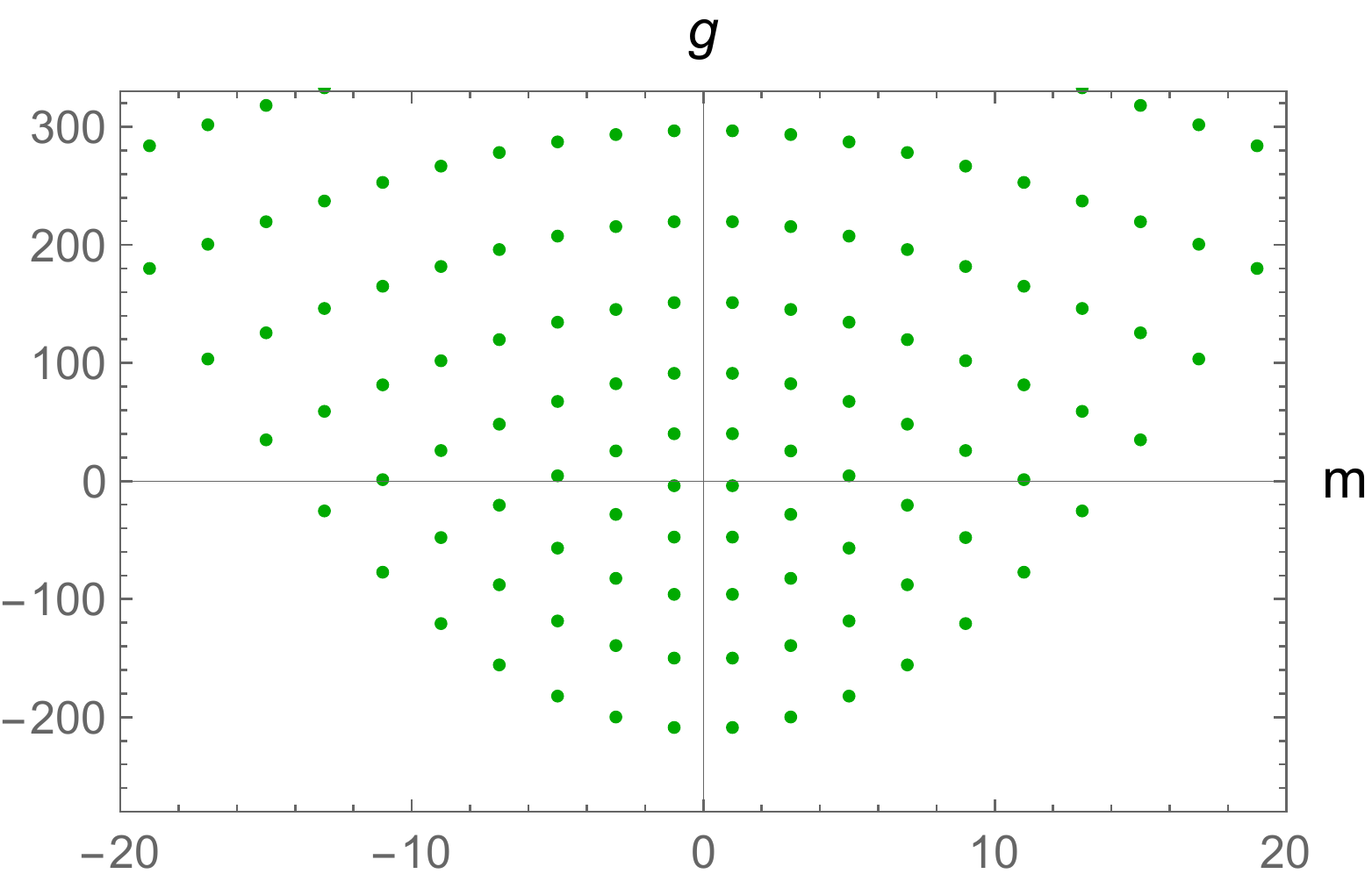} \par
	\end{centering}
	\caption{a) and b) Joint spectrum where $l-m$ is even and $m$ is even/odd
		respectively. c) and d) where $l-m$ is odd and $m$ is even/odd respectively. 	\label{fig:SymmetrySpectrum}}
\end{figure}

\begin{figure}
	\begin{centering}
		\includegraphics[width=5cm]{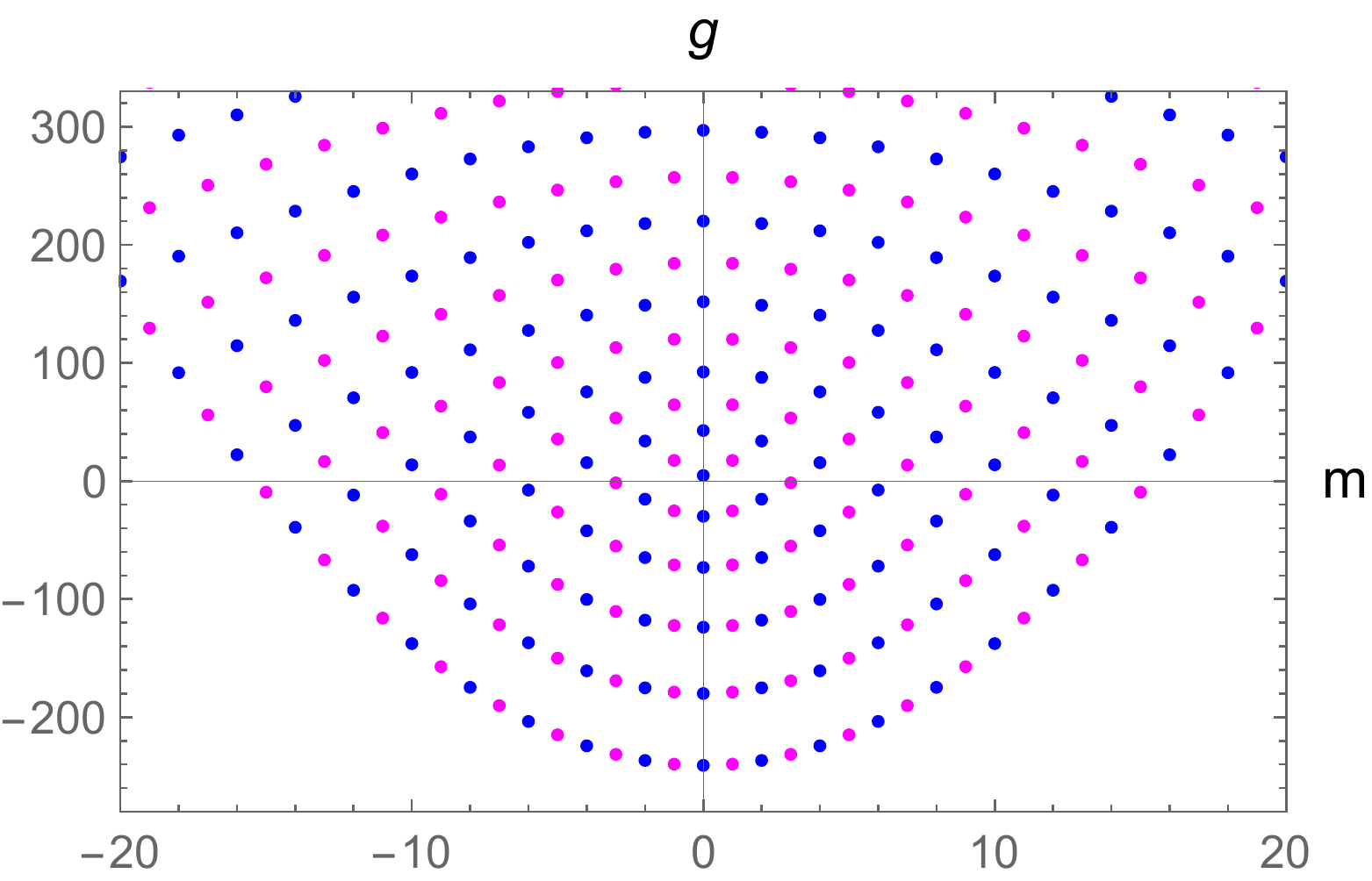}%
		\includegraphics[width=5cm]{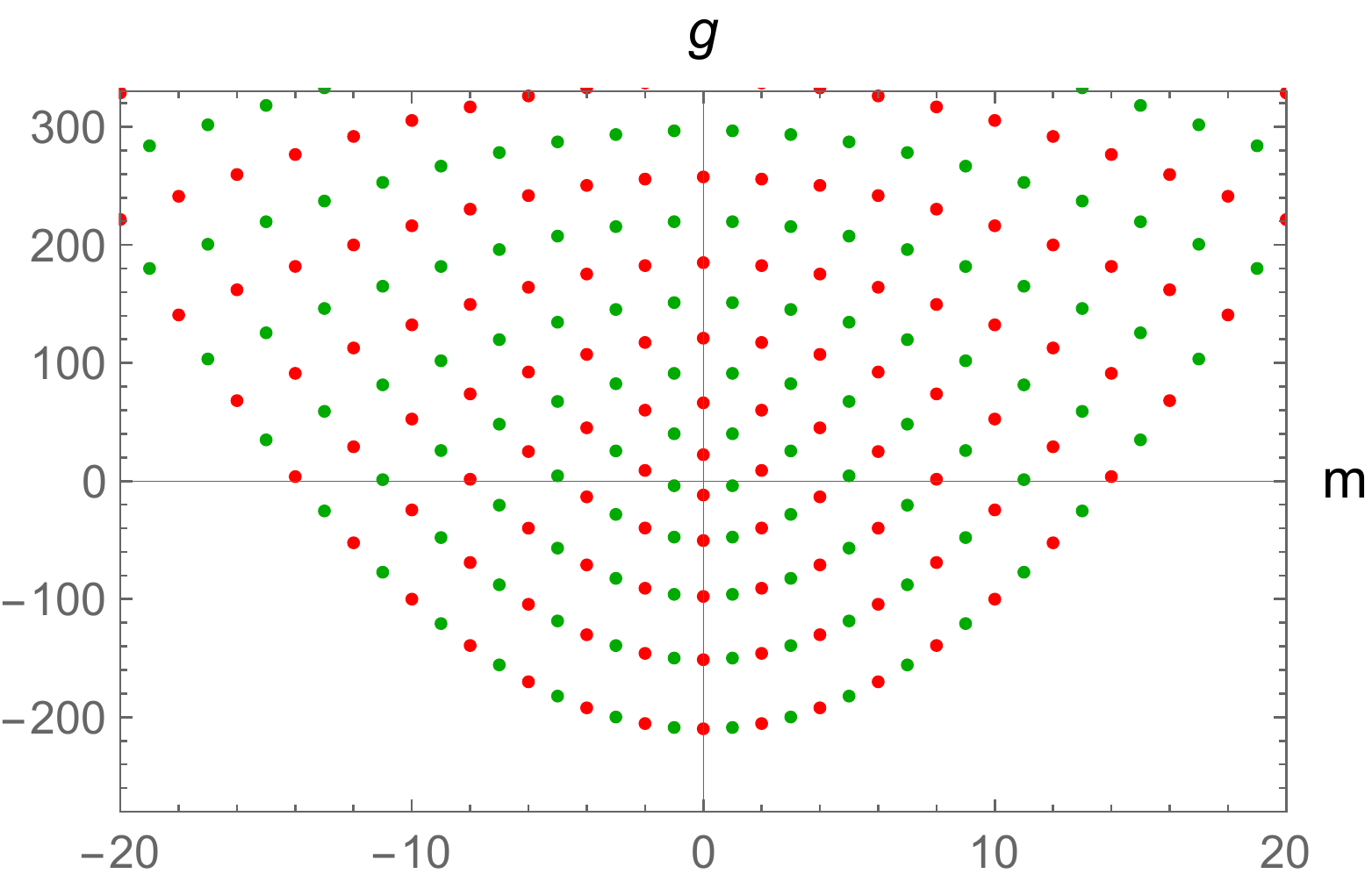}%
		\includegraphics[width=5cm]{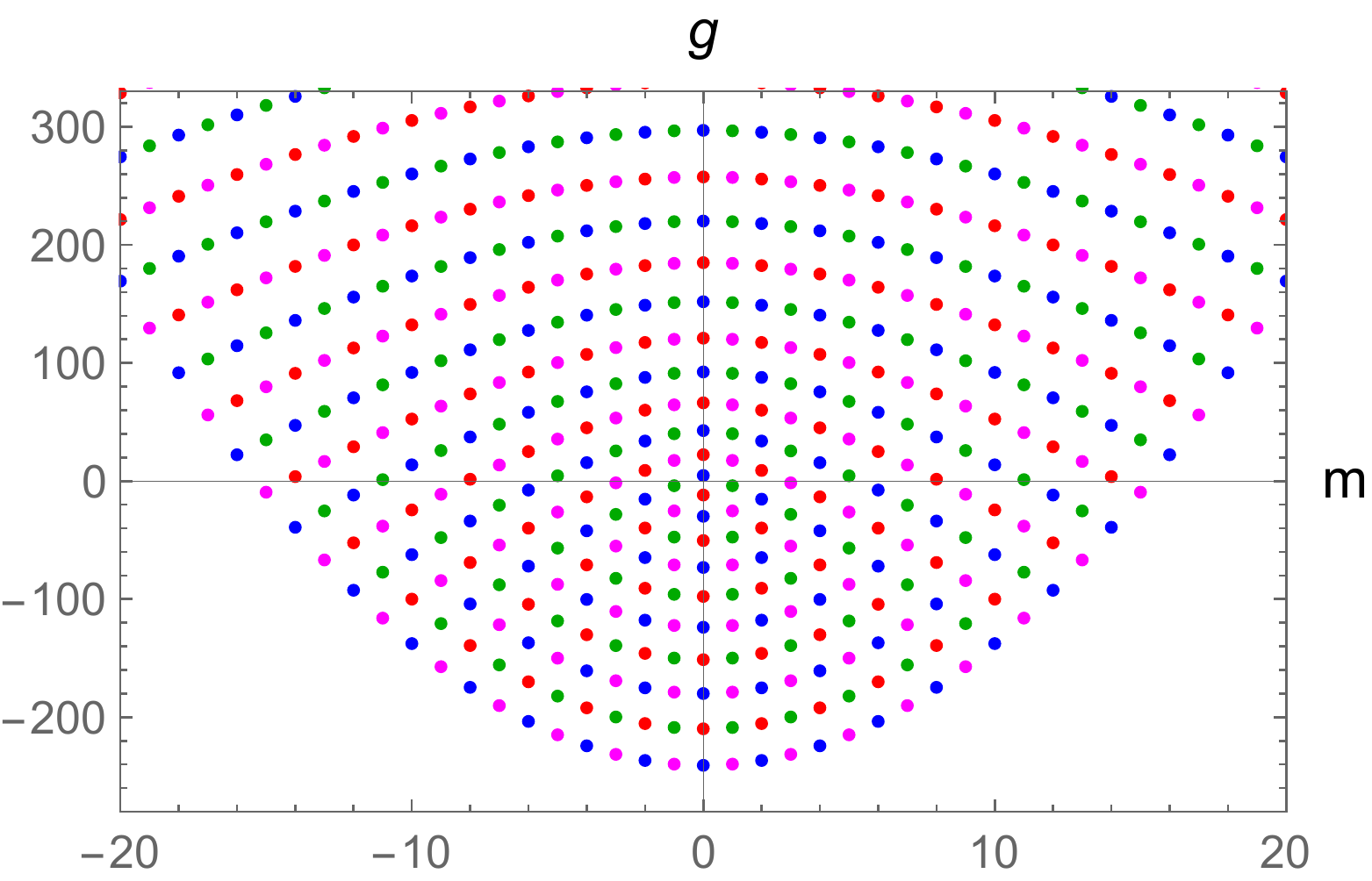}
		\par\end{centering}
	\caption{Parts of the joint spectrum whose eigenfunctions are a) even under $S_{2}$, b) odd
		under $S_{2}$ and c) the complete joint spectrum. The spectra
		shown in a) and b) both have monodromy index $1$. \label{fig:S2 sym} }	
\end{figure}

The joint spectrum can be divided into symmetry classes.
Since \eqref{eq:sep} is even in $\eta$ the eigenfunctions $S_l^m$ are even or odd.
They inherit the symmetry of $P_l^m$ so that $\Ps_l^m$ is even when $l-m$ is even 
and odd when $l-m$ is odd. Accordingly $S_l^m \circ S_1 = (-1)^{l-m} S_l^m$.
Similarly for $\psi_\phi = e^{im\phi}$ it holds that $\psi_\phi \circ S_3 = (-1)^m \psi_\phi $. 
Thus every point in the joint spectrum can be classified according to the parity of $l-m$ and $m$.  This is illustrated in Figure \ref{fig:SymmetrySpectrum} a) through d), where each subfigure contains one quarter of the number of points the full spectrum possesses. Despite this, the unit cell is still deformed in the same way as in Figure \ref{fig:Spheroidal-Eigenvalue} 
and the monodromy index is 2.

It is interesting to note that when selecting states according to their symmetry 
under $S_2 = S_1 \circ S_3$ the monodromy index changes to 1.
Since $Z_l^m$ is the product of $\psi_\phi$ and $S_l^m$ it is 
invariant under $S_2$ if $l-m$ and $m$ are either both even or both odd.
The corresponding joint eigenvalues are shown in Fig.~\ref{fig:S2 sym} left and middle,
and for this selection of joint eigenstates the monodromy is 1. 

The most striking effect of the monodromy is a change in what the symmetry 
of horizontally neighbouring states near the line $m=0$ is. 
Consider Fig.~\ref{fig:S2 sym} left and middle. When $g \gg 0$ the horizontally neighbouring states
have the same symmetry type, while for $g \ll 0$ the symmetry type changes. 
But not only the symmetry type changes, but also the location of states comparing $m=0$ and $m=1$.
For $g \ll 0$ horizontally neighbouring states with $m=0$ and $m=1$ have nearly the same eigenvalue.
By contrast, for $g \gg 0$ consider a state with $m=0$. Now there is no horizontally neighbouring state with $m=1$.
Instead the eigenvalue for a state with $m=1$ is approximately half way between the nearby states with $m=0$. 

We are now going to make these observations precise using \eqref{eq:glmsmallgamma} and \eqref{eq:glmlargegamma}. 
Consider states invariant under $S_2$, hence with even $l-m$ and even $m$,
see Fig.~\ref{fig:S2 sym}, left. Consider the lower end of the figure where $g \ll 0$.
These states are described by \eqref{eq:glmlargegamma}, the asymptotics for large $\gamma$ or small $\hbar$.
The overall ground state has $l = m = 0$. The horizontally neighbouring state with $m=1$ has 
$l = m = 1$ and nearly the same eigenvalue. In particular any state at the lower 
boundary has $l - |m| = 0$.
As noted before, the labelling of states is defined so that it is continuous in the 
limit of spherical harmonics $a \to 0$, and hence the ground state for fixed 
$m$ has $l = |m|$. Using \eqref{eq:glmlargegamma} we find
\[
     g_{l+1}^1 - g_l^0 = 1 + O(1/\gamma), \quad \text{for $g_l^0 \ll 0$.} 
\]
The same analysis holds for Fig.~\ref{fig:S2 sym}, where $l$ in the above formula 
is odd, while in the left figure it is even. Note that the separation of states in the 
vertical direction $g_{l+2}^0 - g_l^0 = 4\gamma + (2l + 3) + O(1/\gamma)$
is of order $\gamma$, and hence we perceive the neighbour in the horizontal 
direction as nearly the same. If we were to present eigenvalues with dimensions
then the difference  in eigenvalue of two horizontally neighbouring states
would be of order $\hbar^2$, while those of two vertically neighbouring states would be $\hbar$.

Now compare this to the situation with large positive eigenvalues  and hence large $l$ near the line $m=0$.
There the state with $m=1$ is approximately equal to the average of neighbouring states with $m=0$.
Using \eqref{eq:glmsmallgamma} we find
\[
    \frac{ g_l^0 + g_{l+2}^0 }{2} - g_{l+1}^1 = 1 + O(\gamma^2), \quad \text{for $g_l^0 \gg 0$}
\]
in the horizontal direction, while in the vertical direction the separation 
is $g_{l+1}^1 - g_l^0 = 2(l+1) + O(\gamma^2 )$. A  comment similar to the previous case 
about the scaling with $\hbar$ applies here.

The previous discussing of neighbouring states was done separately for states that are either invariant under $S_2$ or not.
The reason is that for these subsets the monodromy index is 1. When considering all states
the monodromy index is 2, and its manifestation on the symmetry and labelling of states is different.
In the set of all states in both limits, large positive and large negative $g_l^0$, there is always a horizontally 
neighbouring state with almost the same eigenvalue, see Fig.~\ref{fig:S2 sym}, right.
For large negative $g_l^0$ horizontally neighbouring states with $m=\pm1$ 
have the same symmetry under $S_2$, while for large positive $g_l^0$ horizontally neighbouring states have
opposite symmetry under $S_2$.
A direct consequence of monodromy is the following observation:
for large $l$ such that $g_l^0 \gg 0$ for horizontally neighbouring states $g_l^0 - g_l^1 = O(\gamma^2)$.
For small $l$ such that $g_l^0 \ll 0$ however this difference is not small, $g_l^0 - g_l^1 = 2\gamma + O(1)$. 
This means that states with the same $l$ are not horizontal neighbours, instead the index $l$ needs to 
be increased by 1 when going to the right, then $g_l^0 - g_{l+1}^1 = 1 + O(1/\gamma)$ is small. 
This means that when 
comparing the labelling of states along the line $m=0$ with the line $m=1$ there is a mismatch that 
occurs for small $l$ (negative $g$), while for large $l$ (positive $g$) states are labelled in the natural way.
As already mentioned the fact that this labelling is ``natural'' in the latter case is a choice that was made
in order to have continuity with the labelling in the spherical harmonics limit $a \to 0$.
One could redefine the labelling to be ``natural'' with respect to the Sturm-Liouville problems for 
fixed $m$, then each ground state for fixed $m$ would have the same quantum number.
Then the mismatch in the labelling of horizontal neighbours would appear for states with large eigenvalues $g_l^m$.
The fact that this mismatch cannot be avoided is an expression of the quantum monodromy in the system.

The discussion of monodromy using the asymptotic expansions 
\eqref{eq:glmsmallgamma} and \eqref{eq:glmlargegamma} is enlightning,
but it is somewhat heuristic. If we stay near the line $m=0$ and observe the change in lattice
for small and large $g$ we cannot complete a loop around the focus-focus point, 
because neither formula is valid there. 
If we do complete the loop along the parabolas as indicated in Figure~\ref{fig:Spheroidal-Eigenvalue}
we are stretching the asymptotic expansions to the limit of their validity.
For this reason we are  going to prove existence of monodromy in the semi-classical limit 
by a detailed analysis of the corresponding classically integrable system in section~\ref{moma},
and by appealing to the general theory of quantum monodromy \cite{VuNgoc99}.
It is interesting to note that the general theory  
only makes sense in the semi-classical limit; when  explicit approximate formulas for the 
quantum eigenvalues like  \eqref{eq:glmsmallgamma} and \eqref{eq:glmlargegamma} are known, 
monodromy makes sense as long as there are at least a few eigenvalues $g_l^0 < 0$, so down to say $\gamma = 4$.

%
%

\section{Laplace-Runge-Lenz and C.~Neumann}
\label{LRL}


In this section we will show that the spheroidal harmonics system is symplectomorphic to 
the degenerate C.~Neumann system. The C.~Neumann system is a famous integrable system
that was studied by Jacobi's student Carl Neumann \cite{Neumann1859}, as a prime example of 
separation of variables. It consists of a particle constraint to move on the unit sphere (in any dimension)
under influence of an additional harmonic potential \cite{Moser80a,Moser80,Veselov80,Ratiu81}.
The degenerate case has been studied in \cite{DH03}, and the action variables in the 
general case were analysed in \cite{DRVW99}, also see \cite{DVN05}.
For the quantisation of the C.~Neumann system (in the non-degenerate case) see
\cite{toth93,gurarie95}.

The invariants $\bm{P}$ and $\bm{L} = \bm{Q} \times \bm{P}$ of the free particle are 
of degree 1 and 2 in the original phase space variables. Invariant degree 3 polynomials 
can be formed from them using an analogue of the Laplace-Runge-Lenz vector 
$\bm{A} = \bm{P} \times \bm{L}$.  As in the Kepler problem it is useful to scale with 
the energy: $\bm{K} = \bm{A} |\bm{P}|^{-\alpha}$. 
The Poisson tensor in $\R^9$ with coordinates $\left(\bm{P},\bm{L},\bm{K}\right)$ then is
\begin{equation}
B_\alpha = 
\begin{pmatrix}
\bm{0} & -\hat{\bm{P}} &  \hat{\bm{P}}^2  |\bm{P}|^{-\alpha}   \\ 
-\hat{\bm{P}} & -\hat{\bm{L}} & -\hat{\bm{K}} \\ 
- \hat{\bm{P}}^2  |\bm{P}|^{-\alpha}  & -\hat{\bm{K}} & \hat{\bm{L}}|\bm{P}|^{2(1-\alpha)}
\end{pmatrix} \,.
\end{equation}
In the Kepler problem the idea is to have the bracket between $\bm{L}$ and $\bm{K}$ close,
so there the choice is $\alpha = 1$ so that $|\bm{P}|$ drops out in the lower right corner
and the algebra is $so(4)$.
In our case the choice $\alpha=1$ leads to a realisation of the spheroidal harmonic system on $so(3,1)$, 
but the Hamiltonian $G$ is not smooth when written in terms of $\bm{L}$ and $\bm{K}$, so we do not 
investigate this further.
Instead we are interested to make the bracket between 
$\bm{P}$ and $\bm{K}$ close. To achieve this we need to eliminate $\bm{L}$.
Using standard cross product identities we find 
$\bm{P} \times \bm{A} = -|\bm{P}|^2 \bm{L} + \bm{P} ( \bm{P} \cdot \bm{L})$.
Choosing $\alpha = 2$ thus gives $\bm{P} \times \bm{K} = -\bm{L} + \bm{P}  ( \bm{P} \cdot \bm{L}) |\bm{P}|^{-2}$.
Now fixing the Casimir $\bm{P} \cdot \bm{L} = b$ of $B_\alpha$ allows to eliminate $\bm{L}$ 
and the resulting Poisson structure on $\R^6$ with coordinates $(\bm{P}, \bm{K})$ is
 \begin{equation} \label{eq:BPK}
     B_{P,K} = |\bm{P}|^{-2} \begin{pmatrix}
	\bm{0} & \hat{\bm{P}}^2  \\
	 -\hat{\bm{P}}^2 &  -\hat{\bm{U}}
     \end{pmatrix} , \quad \text{where} \quad \bm{U} = \bm{P} \times \bm{K} - b \bm{P} |\bm{P}|^{-2} 
 \end{equation}
 with Casimirs $\bm{P} \cdot \bm{P}$ and $\bm{P} \cdot \bm{K}$.
Setting the magnetic term $b=0$ and using the identity 
$\bm{P} \bm{P}^t - \hat{\bm{P}}^2 = id \bm{P}\cdot \bm{P}$ we see that
this is the Dirac structure of $T^*S^2$ embedded in $\R^6$ as, e.g., derived in \cite{DH03}.
When considering the Dirac structure of $T^*S^2$ in $\R^6$ we use coordinates 
$\bm{x} = (x_1, x_2, x_3)^t \in S^2$ and momenta $\bm{y} = (y_1, y_2, y_3)^t$ in the tangent 
space of the sphere so that $\bm{x} \cdot \bm{y}  = 0$.
Thus define the Dirac structure $B_D$ of $T^*S^2$ in $\R^6$ as
\begin{equation} \label{eq:Dirac}
    B_D = \begin{pmatrix}
	\bm{0} & -id + \bm{x} \bm{x}^t   |\bm{x}|^{-2} \\
	 id - \bm{x} \bm{x}^t   |\bm{x}|^{-2} & -\widehat{ \bm{x} \times \bm{y}  } |\bm{x}|^{-2}
    \end{pmatrix} 
\end{equation}
with Casimirs $\bm{x} \cdot \bm{x}$ and $\bm{x} \cdot \bm{y}  = 0 $.
Note that the lower left block is the projector to the subspace orthogonal to $\bm{x}$.

\begin{lem} 
Consider the manifold 
$
   M_r = \left\{ (\bm{x}, \bm{y}) \in \R^6 \mid \bm{x} \cdot \bm{x} = r^2, \bm{x} \cdot \bm{y} = 0 \right\} 
$
for $r > 0$.
 The map $\mu: M_r \to M_r$, $(\bm{x}, \bm{y}) \mapsto (\bm{x}, -\bm{x} \times \bm{y})$ 
 is a diffeomorphism with inverse $(\bm{x}, \bm{y}) \mapsto ( \bm{x}, \bm{x} \times \bm{y}/r^2)$.
\end{lem}
\begin{proof}
Composing $\mu$ with $ \mu^{-1}$ and using the vector triple product expansion formula gives
$-\bm{x} \times (\bm{x} \times \bm{y})/r^2 = \bm{y} (\bm{x} \cdot \bm{x} )/r^2 - \bm{x} ( \bm{x} \cdot \bm{y}) /r^2= \bm{y}$.
\end{proof}
Note that for $r=1$ the map $\mu$ of $M_1$ has order 3. If we think of a curve $\bm{x}(t)$ 
on the sphere such that $\bm{y}$ is the tangent vector to the curve then  
 $\mu$ maps the tangent vector to the normal vector. When applied 
a second time $\mu$ maps the normal vector to the binormal vector. 
When applied a third time $\mu$ maps the binormal vector back to the tangent vector.

\begin{prop} \label{todirac}
The map $ (\bm{P}, \bm{L}) \mapsto  (\bm{x}, \bm{y}) = ( \bm{P}, \bm{P}\times\bm{L} |\bm{P}|^{-2})$ is a symplectomorphism
between the co-adjoint orbit of the
 Lie-Poisson structure of $e^*(3)$ in $\R^6$ with variables $\bm{P}, \bm{L}$ given by \eqref{eq:Bpoi} 
to  $T^*S^2$ embedded in $\R^6$ with variables $\bm{x}, \bm{y}$ with Dirac structure given by \eqref{eq:Dirac}.
\end{prop}
\begin{proof}
The Jacobian of the mapping is
\[
M = \begin{pmatrix}
id & 0 \\
-\hat{\bm{L}} |\bm{P}|^{-2} - 2  (\bm{P} \times \bm{L}) \bm{P}^t |\bm{P}|^{-4} & \hat{\bm{P}} |\bm{P}|^{-2} 
\end{pmatrix} \,.
\]
Computing $M B M^t$ gives all blocks but the lower right block of $B_{P,K}$ immediately. 
For this block notice the identity $\hat{\bm{L}} \hat{\bm{P}}^2 + \hat{\bm{P}}^2 \hat{\bm{L}} - \hat{\bm{P}} \hat{\bm{L}} \hat{\bm{P}} = -\hat{\bm{L}} |\bm{P}|^2$
(or in cross-product terms 
$\bm{L} \times (\bm{P} \times (\bm{P} \times \bm{v} ))
+\bm{P} \times (\bm{P} \times (\bm{L} \times \bm{v} ))
-\bm{P} \times (\bm{L} \times (\bm{P} \times \bm{v} ))
=
 \bm{L} \times \bm{v}  |\bm{P}|^2
$ for all $\bm{v} \in \R^3$)
while all other terms vanish because $\bm{P}$ is in the kernel of $\hat{\bm{P}}$. 
Now using the map $\mu$ from the Lemma we see that $\bm{P} = \bm{x}$ and $\bm{L}=-\bm{x}\times\bm{y}$ and this gives the result.
\end{proof}

Having established the  equivalence of the Lie-Poisson structure of $e^*(3)$ of the spherical 
harmonics system with the Dirac structure of $T^*S^2$ the question is what the Hamiltonian 
$G$ becomes when interpreted in these terms.

\begin{thm}
The intergrable spheroidal harmonics system of Theorem~\ref{thm:shis} 
with energy $|\bm{P}|=\sqrt{2E}$ is symplectomorphic to the integrable C.~Neumann system 
of a particle constraint to move on the unit sphere $|\bm{x} | = 1$ with a harmonic potential.
In the coordinates $(\bm{x}, \bm{y})$ on $T^*S^2 \in \R^6$ with the Dirac structure \eqref{eq:Dirac} the Hamiltonian 
of the Neumann system is
\[
    G_N = \frac12 ( y_1^2 + y_2^2 + y_3^2)  - E a^2  ( x_1^2 + x_2^2)
\]
with second integral $L_N = -x_1 y_2 + x_2 y_1$.
 \end{thm}
\begin{proof}
We start with an $\bm{x}$ that is not yet restricted to the unit sphere.
The map from Proposition~\ref{todirac} gives $|\bm{L} | = |\bm{x} | |\bm{y}|$, so that 
the the term $|\bm{L}|^2$ in $G$ becomes $ |\bm{x} |^2 |\bm{y}|^2$. 
Finally we do a symplectic scaling 
to the unit sphere, namely $\bm{x} = c \tilde{\bm{x}}$ and $\bm{y} = \tilde{ \bm{y}}/c$ where $c = \sqrt{2 E}$.
Dropping the tildes and dividing by 2 gives $G_N$.
\end{proof}

In its usual form of the Neumann system has a positive attractive potential. This can be adjusted by 
shifting the potential by the constant term $E a^2 |\bm{x}|^2$, such that the shifted potential is $E a^2 x_3^2$.
To keep the analogy with the spheroidal harmonics integrable system we choose not to do this shift.

Note that while in the spheroidal harmonics system $L_z$ is a coordinate after reduction, 
and this coordinate is a constant of motion, 
in the Neumann system the corresponding integral is again the angular momentum 
$x_1 y_2 - x_2 y_1$ about the third axis but here this is a function of the coordinates $\bm{x}$
and $\bm{y}$. Even when interpreting  $L_z$ as a function of the original coordinates
$\bm{Q}$ and $\bm{P}$ before reduction the difference 
is that then $\bm{P}$ was the momentum, while now after renaming $\bm{P}$ as $\bm{x}$ 
this is the coordinate in configuration space.
When considering the units of the quantities defined we see, however, 
that $\bm{x} =\bm{P}$ does have units of momentum while $\bm{y}  = \bm{P} \times ( \bm{Q} \times \bm{P}) |\bm{P}|^{-2}$ has units of length, 
so that $\bm{x} \times \bm{y}$ does have units of angular momentum, except it has the opposite sign:
$\bm{x} \times \bm{y} = \bm{P} \times ( \bm{P} \times ( \bm{Q} \times \bm{P}) |\bm{P}|^{-2}= -\bm{Q} \times \bm{P}$.

We can introduce spherical coordinates on the unit sphere by
\[
   x_{1}=\sin\theta\cos\phi,  \quad x_{2}=\sin\theta\sin\phi, \quad x_{3}=\cos\theta
\]
which transforms the Hamiltonian $G_N$ to 
\begin{equation}
G_N\left(\theta,\phi,p_{\theta},p_{\phi}\right)=\frac{1}{2}\left(p_{\theta}^{2}+\frac{p_{\phi}^{2}}{\sin^{2}\theta}\right)-E a^2  \sin^{2}\theta
\label{eq:G local}
\end{equation}
where  $p_{\theta}^{2}=\frac{y_{3}^{2}}{\sqrt{1-x_{3}^{2}}}$ and  $p_{\phi}=x_{1}y_{2}-x_{2}y_{1}$ are canonically conjugate momenta to $\theta$ and $\phi$, respectively. 

Thus we see that separation of the (rotationally symmetric) Neumann system in spherical coordinates 
leads to the same Hamiltonian as the prolate spheroidal harmonics system obtained from separation in $\R^3$
in prolate spheroidal coordinates.
A corresponding statement holds for the quantum systems.
Since the phase space $T^{*}S^{2}$ is a
cotangent bundle, Weyl quantization maps the coordinate variables
$(x_{i},y_{i})$ to the operators $(x_{i},\frac{\hbar}{i}\frac{\partial}{\partial x_{i}})$.
The operator corresponding to the Hamiltonian $G_{N}$ is 
\[
2 \hat{G}_{N}=-\hbar^{2}\nabla_{S^{2}}-2E a^{2}\sin^{2}\theta 
\]
which for $\hbar = 1$ can be seen to be the same as \eqref{eq:sep} by making the substitution
$\eta=\cos\theta$.
We close this section by showing the graph of a spheroidal wave function for $m=2, l=4$ and a contour 
plot of the real part of the corresponding spheroidal harmonic $Z_l^m$ on the sphere, along with the spherical harmonic $Y_l^m$ for comparison. Since the potential has its maximum at the poles 
(and its minimum along the equator) the wave function is ``repelled'' from the poles.

\begin{figure}
	\begin{centering}
		\includegraphics[width=8cm]{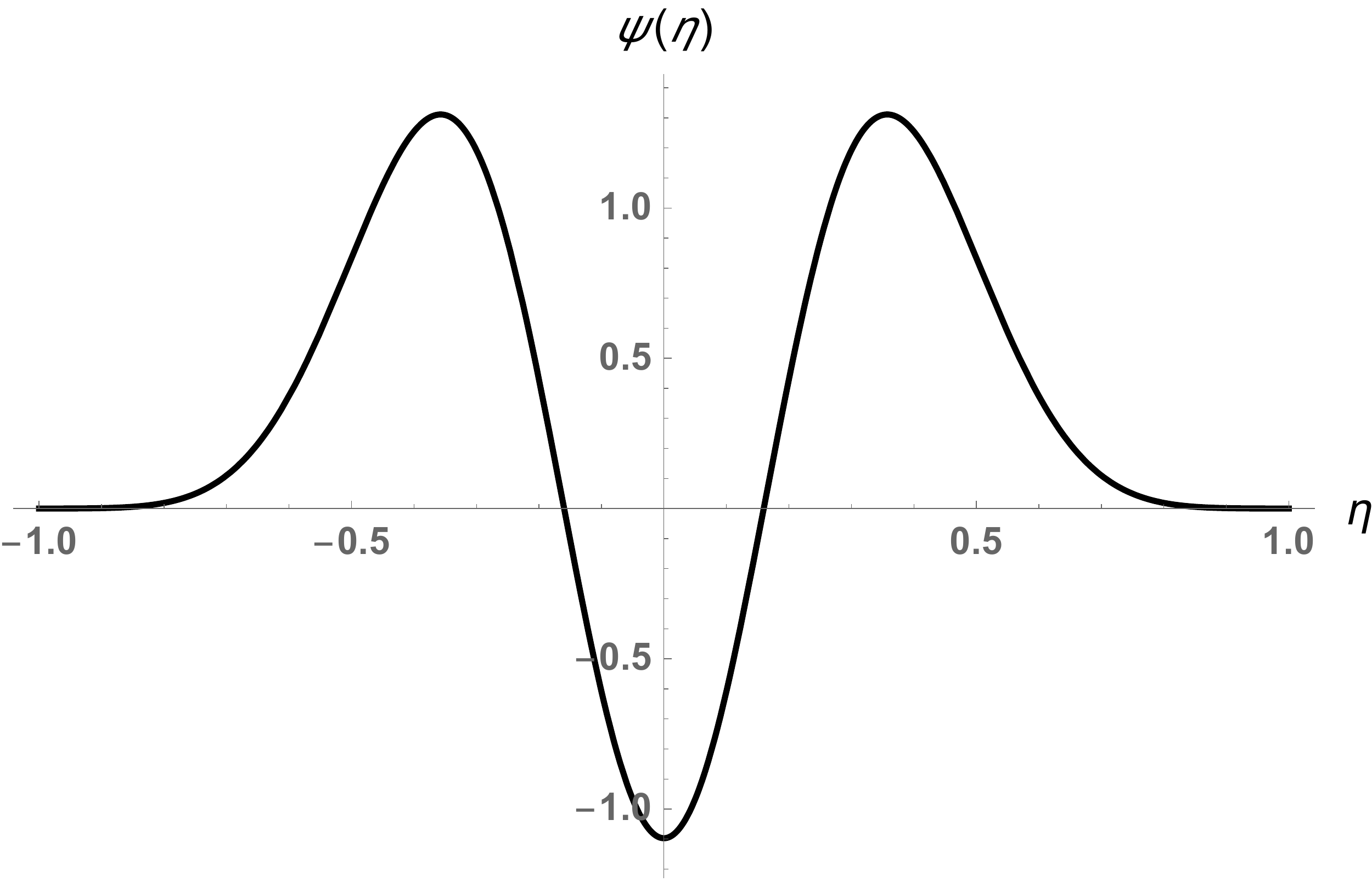}
		\par\end{centering}
	\caption{ Spheroidal wave function with $(n,m,\gamma)=(4,2,20)$.
	}
	 \label{fig:Wave function ellipsoid}
\end{figure}

\begin{figure}
	\begin{centering}
		\includegraphics[width=7cm]{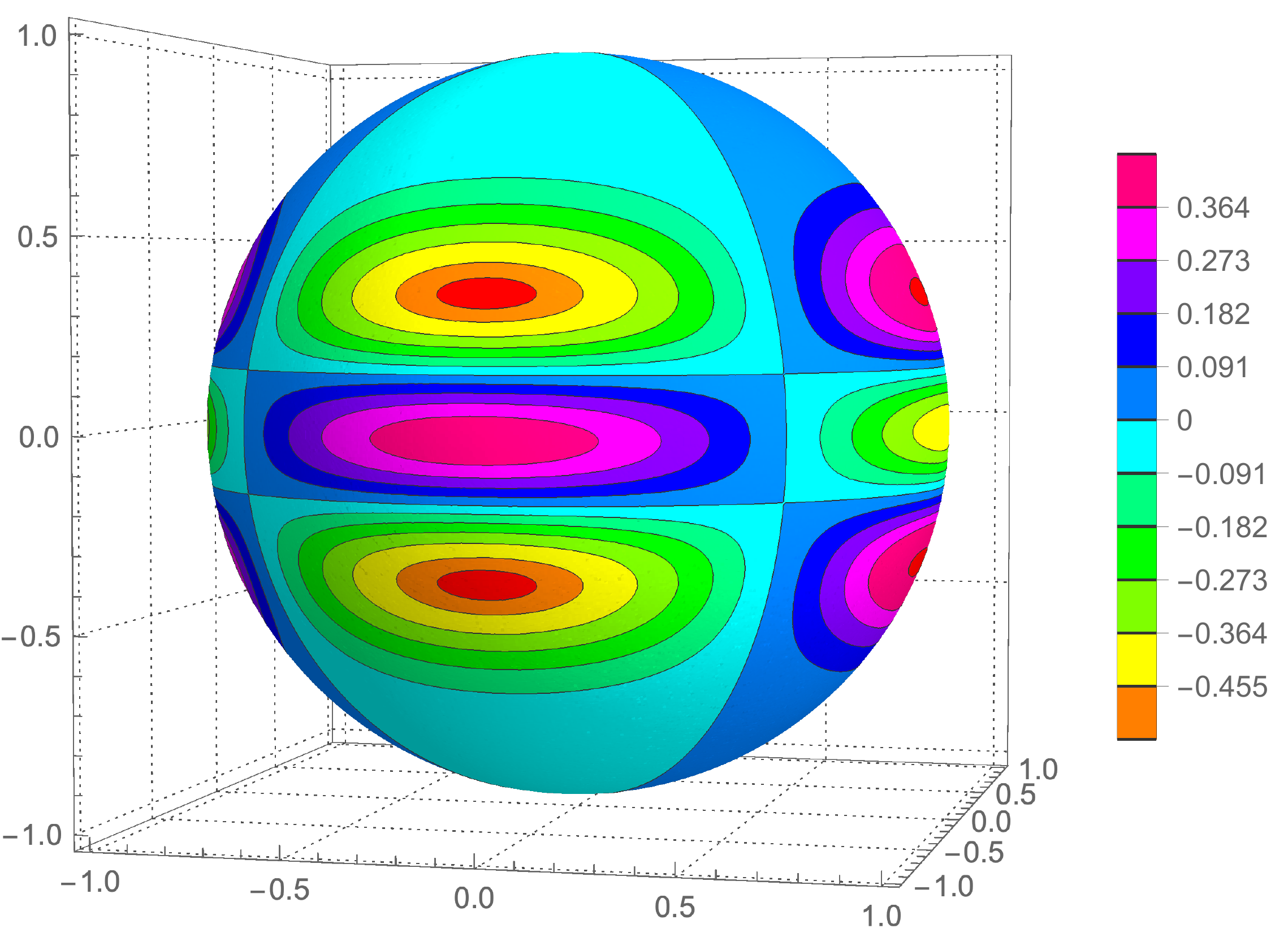}\includegraphics[width=7cm]{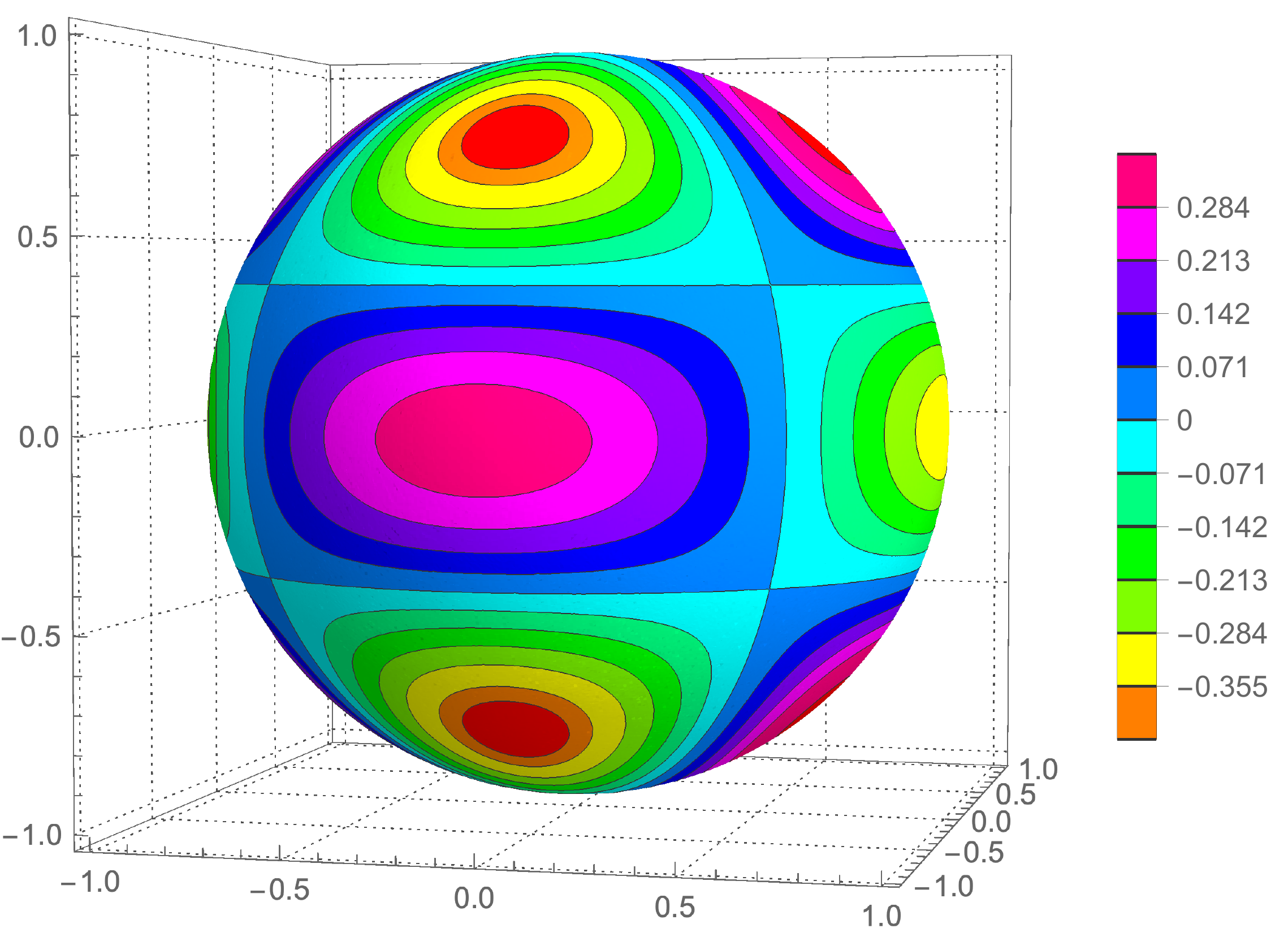}
		\par\end{centering}
	\caption{a)  The spheroidal harmonic $Z_4^2(\theta, \phi)$ with $\gamma = 20$. 
	b) The spherical harmonic $Y_{4}^{2}\left(\theta,\phi\right)$ for comparison.
		\label{fig:WFonsphere}}
\end{figure}

\section{Momentum map of the spheroidal harmonics systems}
\label{moma}

We are now going to analyse the global geometry of the singular Liouville foliation 
of the integrable spheroidal harmonics system. In a number of steps we will prove 
\begin{thm}
The spheroidal harmonics integrable system is a generalised semi-toric system with global $S^1$ action $L_z$.
The momentum map $F=(L_z, G) : T^*S^2 \to \R^2$ has two isolated co-rank 2 critical points 
$\bm{P} = \pm \bm{e}_z \sqrt{2E}$, $\bm{L} = \bm{0}$
and a family of co-rank 1 critical points 
$\bm{P} = \sqrt{2E}( \cos\phi, \sin\phi, 0)^t$, $\bm{L} = \bm{e}_z m$, $\phi \in S^1$, $m \in \R$.
The image of the co-rank 2 critical points is the critical value $(0,0)$, 
which is a non-degenerate focus-focus value and $F^{-1}(0,0)$ is a doubly pinched torus.
The image of the co-rank 1 critical points is the parabola $(m, m^2 - 2Ea^2)$,
points on which are of elliptic-transversal type and $F^{-1}(m, m^2-2Ea^2)$ is a periodic orbit consisting of 
co-rank 1 critical points parametrised by $\phi$. The pre-image of each regular value of $F$ is 
a single torus $\T^2$.
\end{thm}

\begin{figure}
	\begin{centering}
		\includegraphics[width=9cm]{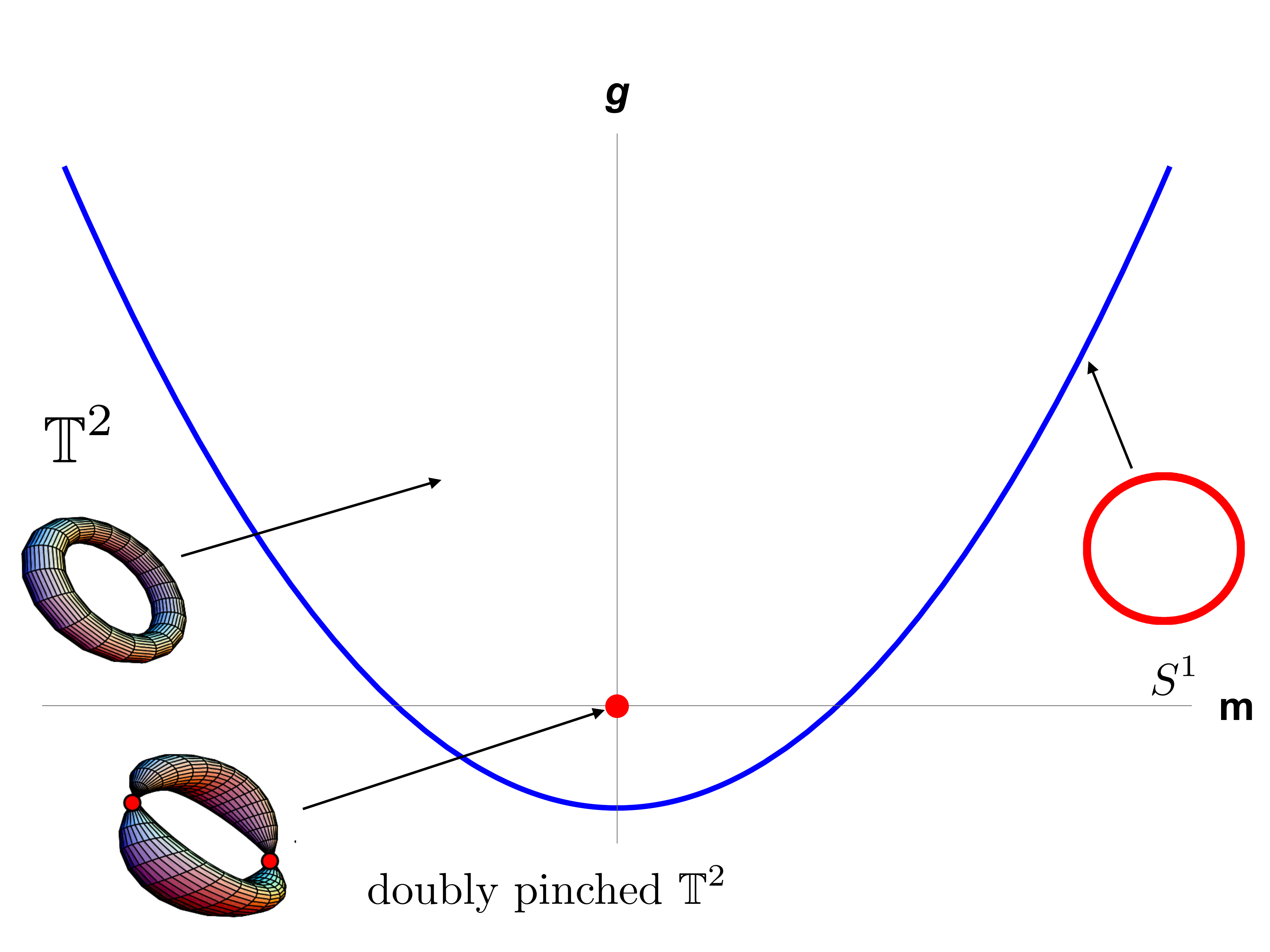}
		\par\end{centering}
	\caption{Bifurcation diagram of the spheroidal harmonics integrable system.}
	\label{fig:momentum map}
\end{figure}
The system will be analysed using singular reduction (using invariants) and regular reduction 
(using global but singular canonical coordinates) and reconstruction to understand the fibres 
of the momentum map. In particular we will show that the focus-focus critical value 
is non-degenerate and hence there is Hamiltonian monodromy in the classical system
invoking \cite{Matveev96,Zung97}. In particular this also implies the existence of quantum monodromy 
in the semiclassical limit as shown in general by San Vu Ngoc in \cite{VuNgoc99}.

We already know a symmetry reduced description \eqref{eqn:Gsep} from 
separation of variables, albeit in singular coordinates. 
Eqn.~\eqref{eqn:Gsep} is connected to the Neumann system 
\eqref{eq:G local} 
via the transformation 
$\eta = \cos\theta$. 
Setting $\hbar = 1$ we have  $l_z =  m$ and arrive at the one degree of freedom Hamiltonian
\begin{equation}
\label{eqn:Gga}
     G(q, p) = (1 - q^2) ( p^2 - \gamma^2) + \frac{m^2}{1 - q^2} \,.
\end{equation}
There is a coordinate singularity at $|q| = 1$. 
The phase portrait of this reduced Hamiltonian is shown in Fig.~\ref{fig:Glevel}.
Away from the singularity there is an equilibrium at the origin with critical value $G(0,0) = m^2 - \gamma^2$.
This gives the line of critical values $g = m^2 - \gamma^2$ in the bifurcation diagram 
Fig.~\ref{fig:momentum map}. The corresponding motion in the original system in Euclidean 
coordinates is a periodic orbit along the equator of the sphere, as already discussed in 
section~\ref{sec:SHIS}.
The parabola of critical values $g = m^2 - \gamma^2$ is also the lower boundary of the joint spectrum 
and is hence shown in Fig.~\ref{fig:Spheroidal-Eigenvalue}.

\begin{figure}
\centering{ \includegraphics[width=5cm]{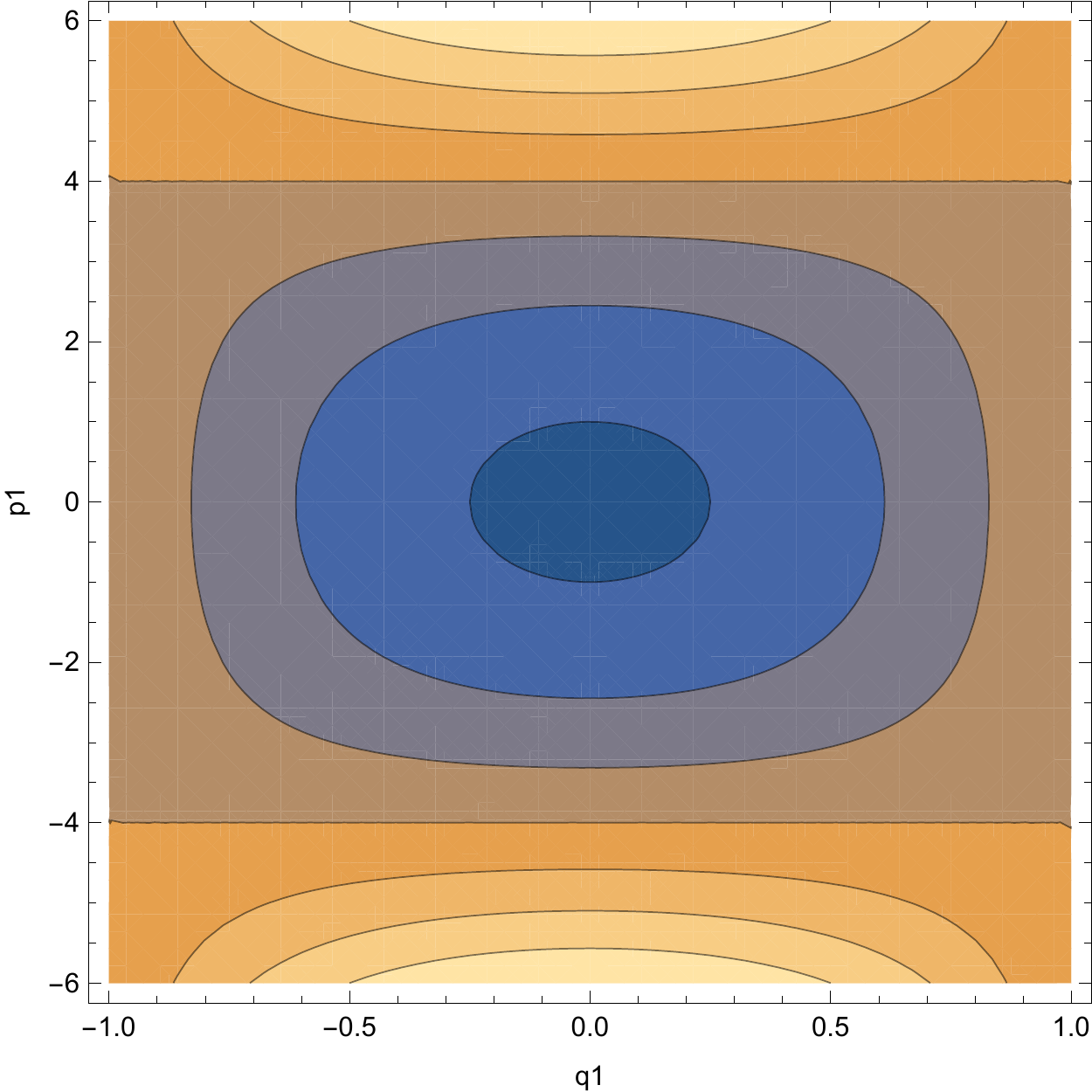} \includegraphics[width=5cm]{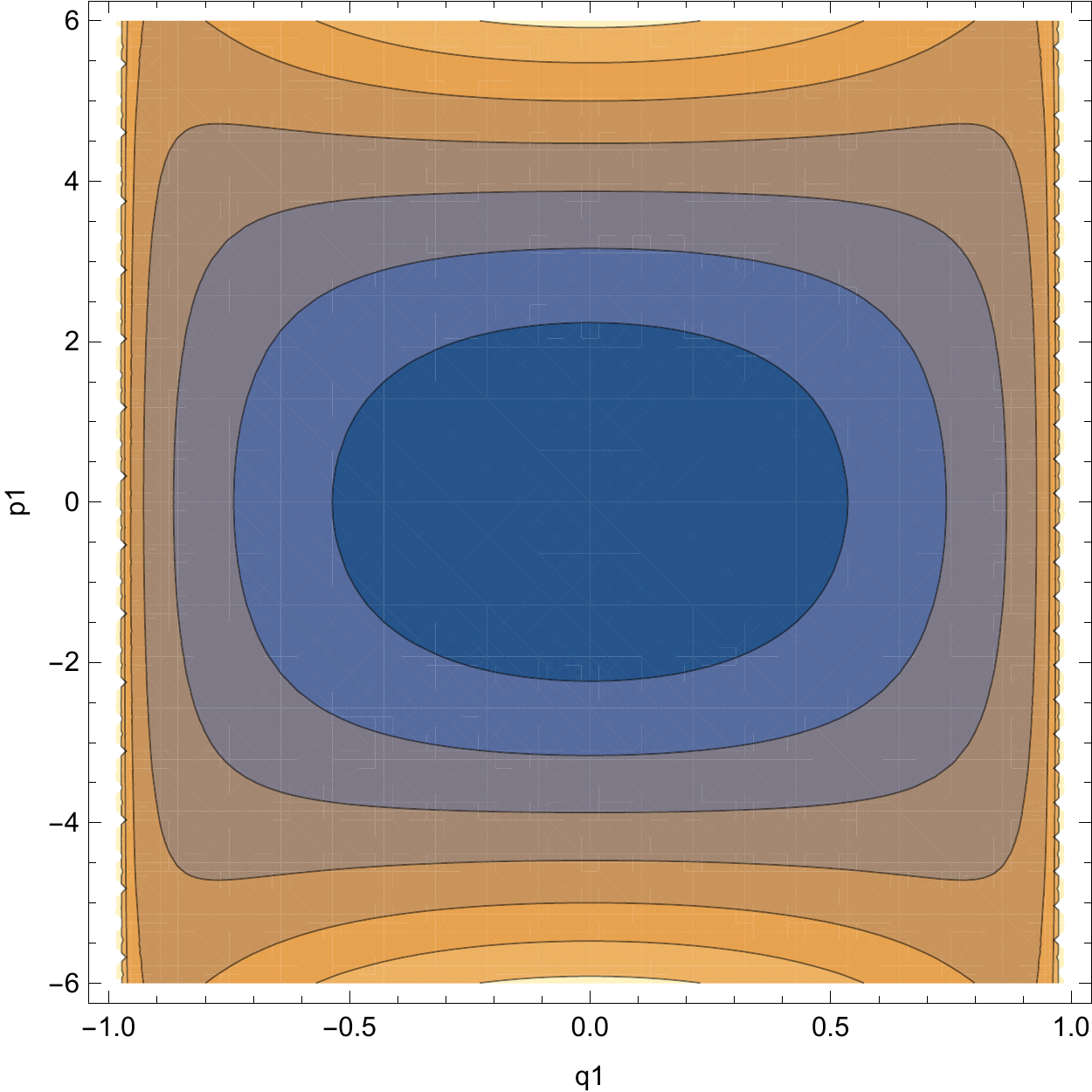} }
\caption{Level lines of $G(q, p)$ for $m = 0$ (left) and $m = 1$ (right), $\gamma = 4$.}
\label{fig:Glevel}
\end{figure}

Since the coordinate system from the separation of variables is singular along 
the $z$-axis we now use singular reduction starting from the global Euclidean 
description in $(\bm{P}, \bm{L}) \in \R^6$ to understand the global dynamics.

\begin{lem} \label{lem:reduction}
Reduction of the spheroidal harmonics system of Theorem~\ref{thm:shis} 
by the global $S^1$ symmetry leads to a Poisson structure in $\R^3$ with coordinates $(b_1, b_2, b_3)$. 
The reduction map $T^*S^2 \to \R^3$ for $|\bm{P}| = \sqrt{ 2E}$ is given by 
\[
    b_1 = \frac{p_z}{\sqrt{2E}} , \qquad
    b_2 =  l_x^2 + l_y^2, \qquad  
    b_3 =  \frac{ l_x p_y - l_y p_x}{ \sqrt{2E} } \,.
\]
with syzygy
\[
       C_3(b_1, b_2, b_3) =  (1 - b_1^2) b_2 - b_1^2 m^2 - b_3^2 = 0 \,.
\]
The Poisson tensor is $\widehat{ \nabla C_3 }$.
\end{lem}
\begin{proof}
The global $S^1$ action $L_z$ as a Hamiltonian with respect to the Poisson structure $B$ 
generates a rotation in the first two components of $\bm{P}$ and $\bm{L}$ and fixes the third component, 
see \eqref{eq:S1flow}.
Thus $p_z$ and $l_z$ are invariant under this symmetry.
Introducing $p_w = p_x + i p_y$ and $l_w = l_x + i l_y$ the $S^1$ action is  multiplication of $p_w$ and $l_w$ by $e^{i\phi}$. 
Any polynomial of $p_z$ and $l_z$ is also invariant.
Additional quadratic polynomial invariants are $|p_w|^2$, $|l_w|^2$ and the real and imaginary part of $p_w \bar l_w$.  
All other polynomial invariants are functions of these 6 invariants, 2 linear and 4 quadratic. 
The Casimirs of the Poisson structure $B$ expressed in these invariants read $|p_w|^2 + p_z^2 = 2E$ 
and $\Re( p_w \bar l_w) + p_z l_z = 0$ and can be used to eliminate 
$|p_w|^2$ and $\Re( p_w \bar l_w)$ wherever they appear.
As before we set $l_z = m$ where $m$ is now considered as a parameter.
In addition we scale the momentum with $\sqrt{2E}$ as for the transformation 
to the Neumann system.  
The remaining invariants are denoted by $b_i$ where $b_2 =  |l_w|^2$
and $b_3 =  \frac{ \Im( p_w \bar l_w) }{ \sqrt{2E} } $. 
This gives the stated reduction map. 
The invariants satisfy $|b_1| \le 1$ and $b_2 \ge 0$ by construction.
The identity  $\Re( p_w \bar l_w)^2 + \Im( p_w \bar l_w)^2 = | p_w \bar l_w|^2 = |p_w|^2 |l_w|^2$ 
rewritten in terms of the invariants gives $C_3 = 0$.
A fundamental property of invariants is that their Poisson bracket is again an invariant.
By using the original Poisson structure $B$ in the original variables $(\bm{P}, \bm{L})$ 
one can verify that
\[
  \{ b_1, b_2 \} = 2 b_3, \quad
  \{ b_1, b_3 \} = 1 - b_1^2, \quad
  \{ b_2, b_3 \} = 2 b_1 m^2 + 2 b_1 b_2 \,.
\]
The right hand sides are given by the derivatives $\partial C_3 / \partial b_i$,
such that the reduced Poisson structure is $\widehat{ \nabla C_3}$ as claimed.
By construction then $C_3$ is a Casimir of the reduced Poisson structure.
Since this encodes an identity between invariants (a so-called syzygy) the value
of $C_3$ must be zero. 
\end{proof}

\begin{figure}
\noindent \begin{centering}
\includegraphics[height=5cm]{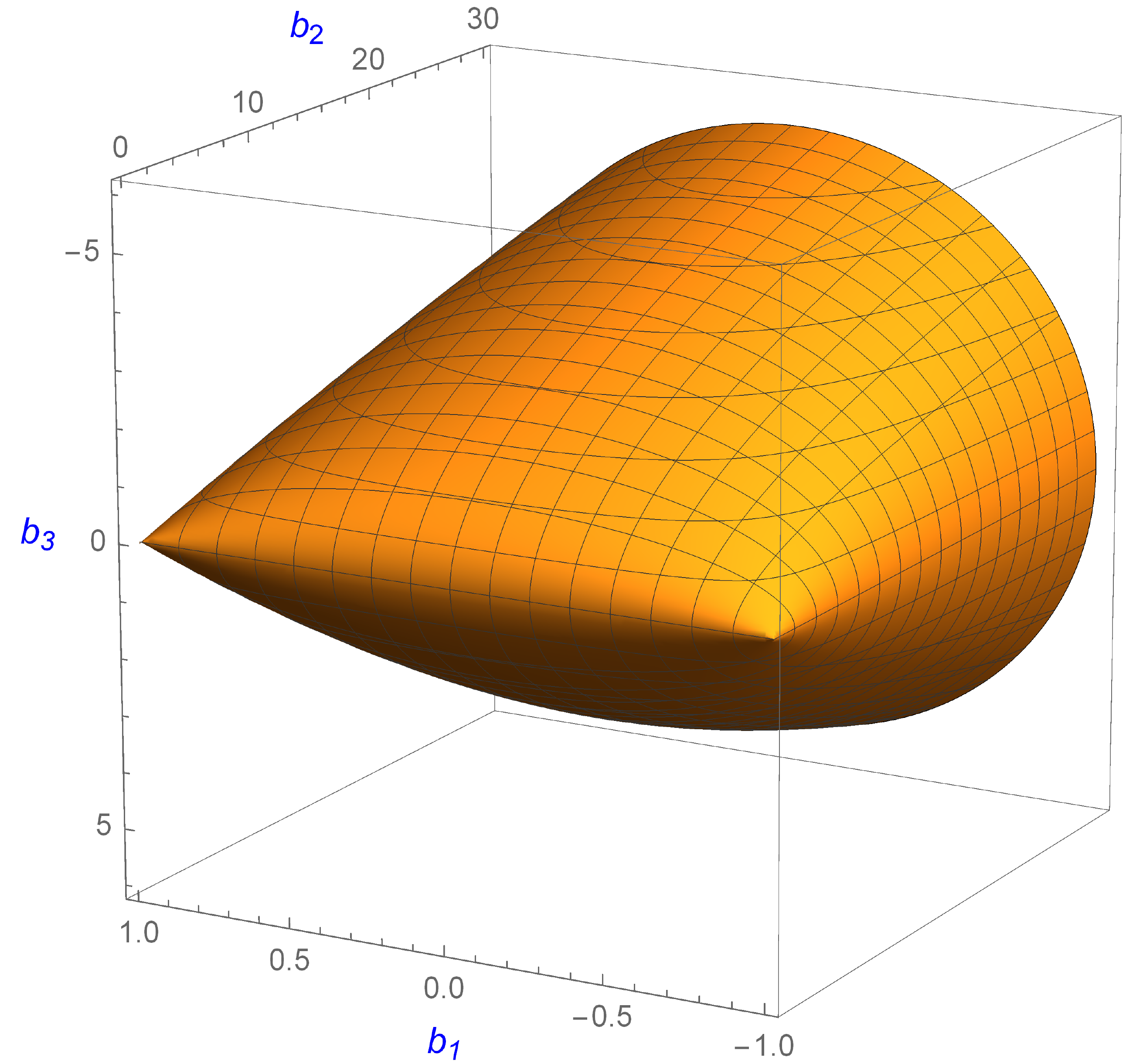}
\includegraphics[height=5cm]{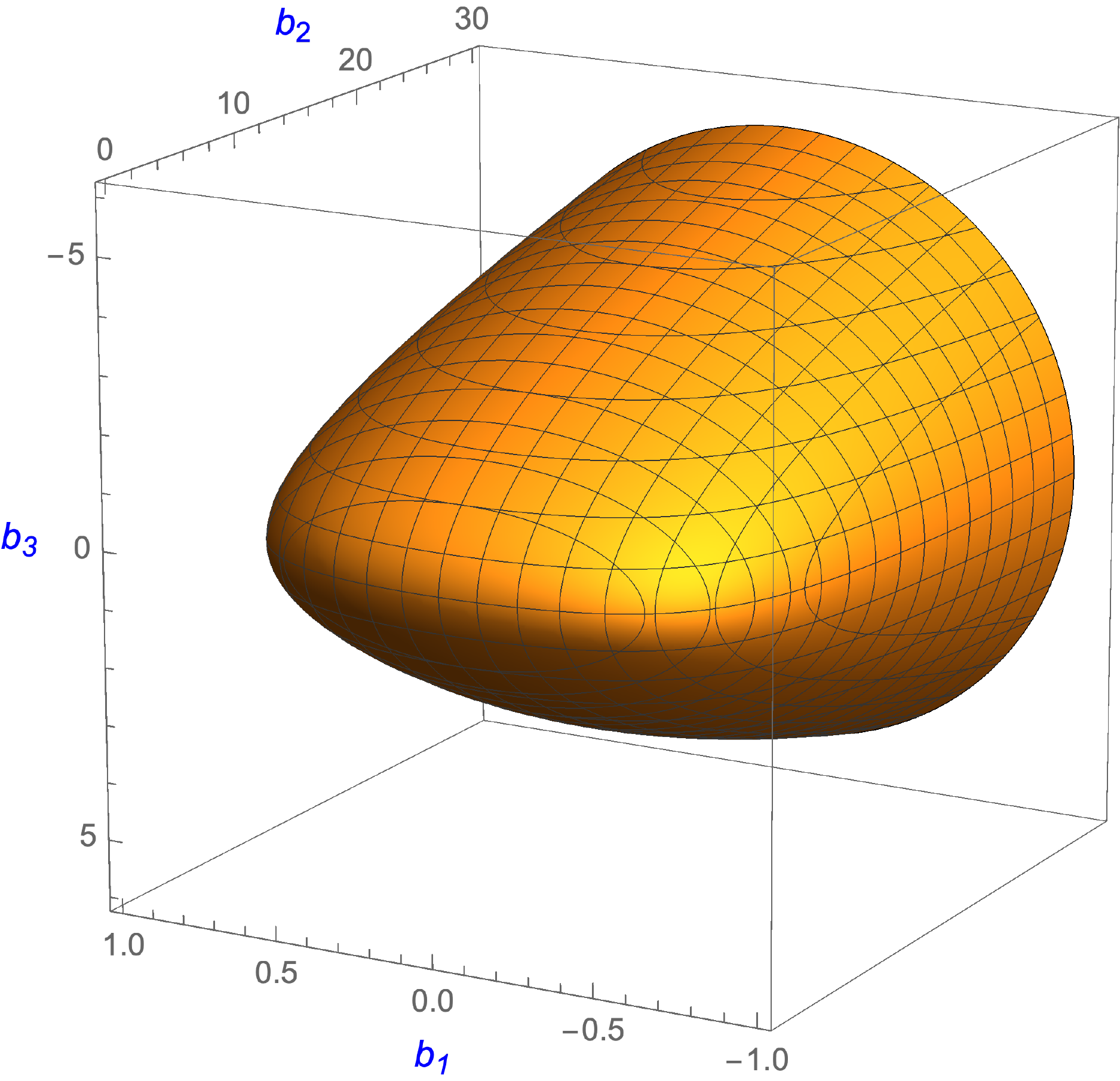}
\par\end{centering}
\caption{a) The singular reduced phase space $P_{m=0}$ 
with two singular points at $(b_1, b_2, b_3) =(\pm1, 0, 0)$;
b) A regular reduced phase space $P_{m=2}$ with non zero $m$. \label{fig:-pillow} }
\end{figure}


The invariants can of course also be written in the 
coordinates $(\bm{x}, \bm{y})$ of the Neumann system on the unit sphere
where they look more natural as
\[
b_1 =  x_3, \qquad
b_2 =  y_1^2 + y_2^2, \qquad  
b_3 = y_1 x_2 - y_2 x_1 \,.
\]
The points $\bm{P}=(0,0,\pm\sqrt{2E})$ and $\bm{L}=(0,0,0)$ are fixed under rotations about the third axis.
Hence the global $S^1$ action has fixed points and the symmetry reduced phase space is not in general a smooth manifold. 
This is the reason that we are using singular reduction. This fixed point occurs for $l_z = m = 0$ and 
its image under the reduction map is $(\pm1, 0, 0)$. We now verify that these are exactly the singular 
points of the reduced phase space.
\begin{lem}
The reduced phase space $P_{m} = \{ (b_1, b_2, b_3) \mid C_3 = 0, b_2 \ge 0, b_1^2 \le 1 \}$
is a smooth surface for $m \not = 0$ and a singular semi-algebraic variety with 
two conical singularities at $(b_1, b_2, b_3) = ( \pm 1, 0, 0)$ for $m=0$.
\end{lem}
\begin{proof}
The reduced phase space is the subset of $\R^3$ with coordinates 
$b_1, b_2, b_3$ for which the syzygy Casimir is satisfied, $C_3 = 0$,
and in addition the inequalities $b_2 \ge 0$ and $b_1^2 \le 1$ hold.
Singular points occur when $\partial C_3 / \partial b_i = 0$ which 
implies $b_3 = 0$, $b_1 = \pm 1$ and $b_2 = -m^2$, which is 
only possible for $m = b_2 = 0$. Thus for $m=0$ the 
variety $\{ C_3 = 0 \} $ is not a smooth manifold, 
but has two singular points at $(\pm1, 0, 0)$, see Figure \ref{fig:-pillow}.
For $m \not = 0$ it is a smooth manifold. The inequalities select one
connected component.
\end{proof}

The next step is the analysis of the dynamics of the reduced system.
We write the Hamiltonian $G$ of \eqref{eqn:Gorg} in terms of invariants as
\begin{equation} \label{eqn:invariantG}
     G (b_1,b_2,b_3) =  b_2 + m^2 - \gamma^2( 1 - b_1^2) 
\end{equation}
using $l_z = m$ and $\gamma = 2 E a^2$ with $\hbar = 1$.
The trajectories of the reduced system are given by the
intersection of the  reduced ``energy surface" $\{ G = g \}$ with reduced phase space $P_m$.
This leads to the description of the image of the momentum map $(L_z, G)$, see Fig.~\ref{fig:momentum map}.

%
%
%
\begin{figure}
\noindent \begin{centering}
\includegraphics[width=5cm]{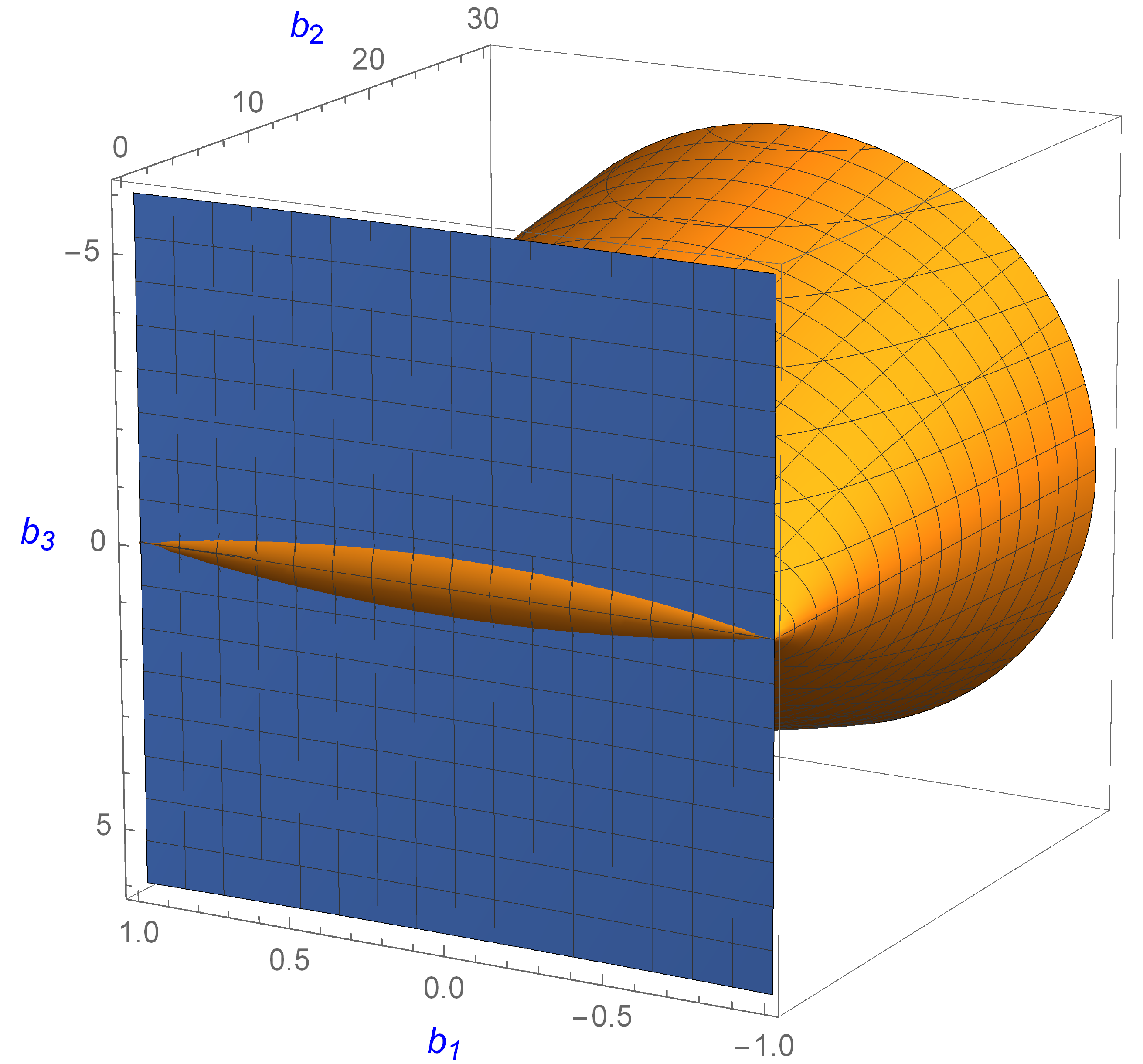}
\includegraphics[width=5cm]{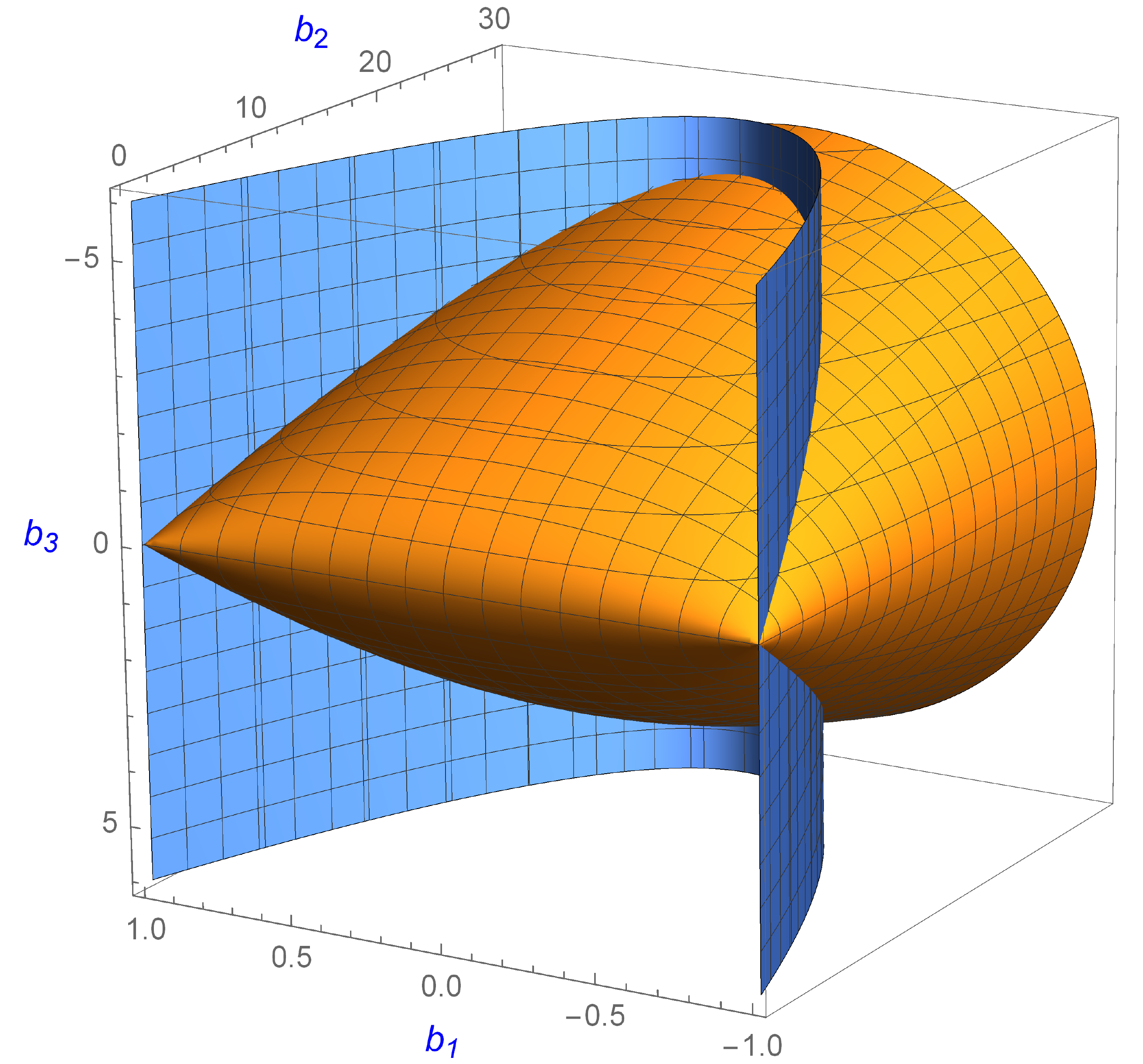}
\par\end{centering}
\caption{Separatrix connecting the singular points. 
It is given by the intersection of the singular reduced phase space $P_0$ (yellow) 
with the energy surface $\{ G = 0 \}$ (blue) for  $\gamma = 0.5$ (left) and $\gamma=5$ (right).}  
\label{fig:tori}
\end{figure}

\begin{lem} \label{lem:critval}
The set of critical values of the energy-momentum map $(L_z, G)$ consists of an 
isolated point at the origin $(0,0)$ and the parabola $g = m^2 - \gamma^2$.
The corresponding critical points are $(\pm 1, 0, 0)$ and $(0,0,0)$, respectively.
The sepatratrices connecting $(\pm 1, 0, 0)$ are the parabolic arcs
$(b_1, b_2, b_3) = \left(b_1, \gamma^2(1-b_1^2),\pm \gamma (1-b_1^2) \right)$.
\end{lem}
\begin{proof}
In general a tangency between the reduced phase space  $P_m$ and the 
parabolic cylinder $\{ G = g \}$ occurs when their gradients are parallel, 
which implies $b_3 = 0$ and either $b_1 = 0$ or 
$b_2 = -m^2 - \gamma^2( 1 - b_1^2)$. 
Since $b_2 \ge 0$ the latter implies $b_1 = \pm1$ and $m = 0$.
These are two isolated critical points at $(\pm 1, 0, 0)$ 
both with isolated critical value $(m,g) = (0,0)$.
The preimage of this critical value in the reduced system is given by the 
intersection of the singular reduced phase space $P_0$ with the reduced energy surface 
$\{ G(b_1, b_2, b_3) = 0 \}$. Solving $G = 0$ with $m=0$ gives the equation for $b_2$.
Inserting into $C_3 = 0$ and extracting a square root gives the equation for $b_3$.
See Fig.~\ref{fig:tori} 

In the other case of parallel gradients with $b_1 = 0$ the Casimir $C_3 = 0$ implies $b_2 = 0$ as well,
so that the critical point is $(0,0,0)$ with corresponding family of critical values 
$(m,g) = (m, m^2 - \gamma^2)$. 
%
%
All points in the $(m,g)$ plane above the parabola $g = m^2 - \gamma^2$ 
with the exception of the origin are regular values.
For each regular value the intersection of $P_m$ and $ \{ G = 0 \}$ is a single curve 
diffeomorphic to $S^1$. 
These intersections can also be seen as the level lines of $G(q,p)$ as shown in 
Fig.~\ref{fig:Glevel} (right).
\end{proof}

The final step in the analysis of the classical dynamics is the reconstruction,
which leads to a description of the invariant sets of the dynamics in the 
original coordinates $(\bm{P}, \bm{L})$. The reduction map 
of Lemma~\ref{lem:reduction} is a projection from the 4-dimensional space 
$T^*S^2 \subset \R^6$ to $\R^3$.
\begin{lem} \label{lem:reconstruct}
For given $b_1, b_2, b_3$ points in the preimage of the reduction map are given by 
\[
    \bm{P} = \sqrt{2E}\left( \sqrt{ 1 - b_1^2} \cos u, \sqrt{ 1  - b_1^2} \sin u , b_1 \right), \quad
    \bm{L} = \left(\sqrt{b_2} \cos v ,\sqrt{b_2} \sin v, m\right)
\]
where $u - v = \arg( -b_1  m + i b_3)$.
The $S^1$ action increases both $u$ and $v$ by $\phi$ and leaves the difference $u-v$ invariant.
\end{lem}
\begin{proof}
In Lemma~\ref{lem:reduction} we already noted that the $S^1$ action is most easily
described by multiplication with $e^{i\phi}$ after introducing the complex variables $p_w = p_x + i p_y$ and $l_w = l_x + i l_y$.
By definition $b_2$ is the modulus squared of $l_w$ and $b_1$ is the normalised size of $p_z$, 
such that $|p_w|^2 = 2E - p_z^2 = 2E( 1 - b_1^2)$.
Thus there are angles $u$ and $v$ such $e^{i\phi} p_w = \sqrt{ 2E(1 - b_1^2)} e^{i u} $
and $e^{i\phi} l_w = \sqrt{ b_2} e^{i v}$. For given $b_1, b_2, b_3$ the arguments $u$ and $v$
are related. On the one hand from Lemma~\ref{lem:reduction}  we have $\Re(p_w \bar l_w) = - p_z l_z$
and $\Im(p_w \bar l_w) = \sqrt{2E} b_3$, such that $p_w \bar l_w = \sqrt{2E}( b_1 m + b_3)$.
On the  other hand $p_w \bar l_w = \sqrt{2E} \sqrt{ 1 - b_1^2} \sqrt{b_2} e^{i(u-v)}$, and hence the result.
At the singular point $(\pm 1, 0, 0)$ the angles $u$ and $v$ are undefined, but this is the fixed point of the 
$S^1$ action, so the preimage of each of these points is just a single point each, instead of a circle each.
\end{proof}
It is interesting to note that these formulas can be directly expressed in terms of the original 
separating variables. In particular both, $p_w$ and $l_w$ when expressed in terms of 
$(\xi, \eta, \phi, p_\xi, p_\eta, p_\phi)$ after cotangent lift of the definition \eqref{eq:To prolate xi eta} of spheroidal 
coordinates can be written as $p_w = e^{i\phi} p_{w0}$ and $l_w = e^{i\phi} l_{w0}$
where $p_{w0}$ and $l_{w0}$ are independent of $\phi$. This leads to formulas for 
$b_1,b_2,b_3$ in terms of the separating variables. One subtlety here is that in such 
formulae the value of $E$ is not fixed, but is determined by the values of $\xi, \eta, p_\xi, p_\eta$, while $l_z = p_\phi = m$, as always.
The difference in the reconstruction formula is that there $\xi$ and $p_\xi$ have been eliminated.

Symplectic coordinates on the reduced phase space can be introduced by 
\[
    (q, p) = \left(   b_1  ,  \frac{ b_3}{1 - b_1^2}  \right) \,.
\]
It is easy to check that these functions satisfy $\{ q, p \} = 1$, and that 
they reduce the Poisson structure $\widehat{ \nabla C_3 }$
 in $\R^3$ to the standard symplectic structure in $\R^2$.
Using the Casimir to express $b_2$ as a function of $(q,p)$ the 
Hamiltonian $G$ in \eqref{eqn:invariantG} can be turned into the form \eqref{eqn:Gga}.
Of course reintroducing symplectic coordinates also reintroduces the coordinate singularity.

However, notice that through the chain of transformations we have arrived again at the 
separated Hamiltonian function $G$ albeit evaluated in different coordinates.
Originally the separation gave a function $G(q,p)$ where either $(q,p) = (\eta, p_\eta)$ 
or $(q,p) = (\xi, p_\xi)$. The variables $(q,p)$ just introduced as a function of $b_i$ however 
set $q = p_z/\sqrt{2E}$ and $p = \sqrt{2E} (\bm{P} \times \bm{L})_z/( p_x^2 + p_y^2)$.

In order to classify the critical point corresponding to the critical values 
the dynamics needs to be analysed in full phase space. 
First we show that the preimage of the isolated critical value 
$(0,0)$ of the momentum map $(L_z, G)$ is a doubly pinched torus, 
and then we will show that it is a non-degenerate focus-focus critical value.

%
%
%
%

\begin{lem} \label{lem:ff}
The preimage of the critical value $(0,l0)$ of the prolate spheroidal harmonics system 
is a doubly pinched torus with $l_z = 0$ in the phase space $T^{*}S^{2}$ parametrised by 
$p_z$ and $\phi$ as
\begin{align*}
\begin{pmatrix}p_{x}\\
p_{y}\\
l_{x}\\
l_{y}
\end{pmatrix} & =\sqrt{2E-p_{z}^{2}}\begin{pmatrix}1 & 0\\
0 & 1\\
0 & \pm a\\
\mp a & 0
\end{pmatrix}\begin{pmatrix}\cos \phi \\
\sin \phi
\end{pmatrix}.
\end{align*}
\end{lem}
\begin{proof}
Combining the parabolic arcs from Lemma~\ref{lem:critval} with 
the reconstruction formula Lemma~\ref{lem:reconstruct} for the 
case $g = m = 0$ gives the result. 
We have $\Re( p_w \bar l_w) = 0$ since $m = 0$ and hence 
$u - v = \pm \pi/2$ 
where the plus sign correspond to the upper parabolic arc with $b_3 \ge 0$ 
and the minus sign to the lower arc with $b_3 \le 0$.
\end{proof}

This Lemma gives  a parametrisation of the doubly pinched torus in phase space. 
For the spheroidal harmonics system it is even possible to describe the dynamics 
on this doubly pinched torus in terms of simple formulas.
Consider the local symplectic coordinates $G(q,p)$.
When $m=0$ then $G=0$ implies either $q = \pm1$ or $p = \pm \gamma$. We choose the second condition to stay 
away from the critical point. Hamilton's equations then say that $p = \pm \gamma$ is constant, 
as can be seen in Figure~\ref{fig:Glevel}. The remaining ODE for $q$ can be solved
 to give $q(t) = \tanh( \pm 2 t \gamma - c)$, which is the connection from the north-pole 
 to the south-pole of the sphere, or vice versa, depending on the sign of $p = \pm \gamma$.
 The dynamics of $\phi$ is trivial, since $\dot \phi = -\partial G(q,p)/\partial m = 0$ for $m=0$.

%

%

%
%
%


\begin{lem}
The critical value $(0,0)$ of the momentum map $(L_z, G) : T^*S^2 \to \R^2$ 
is a non-degenerate focus-focus value. 
The critical values $(m, m^2 - \gamma^2)$ are non-degenerate values of elliptic-transversal type.
\end{lem}
\begin{proof}
At a critical point of the map $(L_z, G)$ the flows (in the original coordinates) generated by
$G$ and $L_{z}$ are parallel:
\begin{equation}
\alpha B\nabla G+\beta B\nabla L_{z}=\bm{0}, \quad   \beta\in\R\backslash\left\{ 0\right\} .
\label{eq:flow paralle}
\end{equation}
The vector fields are given by \eqref{eqn:GVF} and \eqref{eq:S1flow}, and since 
the the former is non-vanishing for $E>0$ we can set $\alpha = 1$.

Critical points of the form $\bm{P}=(0,0, p_z)$ and $\bm{L}=\left(0,0,0\right)$
with $\beta$ arbitrary have the critical values $(0,0)$.
Critical points of the form  $\bm{P}=\left(p_{x},p_{y},0\right)$ and $\bm{L}=\left(0,0,m\right)$
with $\beta=m$ have the critical value $(m, m^2 - \gamma^2)$.

The essential object for the classification of critical values and non-degeneracy 
are the eigenvalues of  the Jacobian $\partial_{\bm{P},\bm{L}}\left(B\nabla G+\beta B\nabla L_{z}\right)$
at these critical points. Two of the six eigenvalues are always zero;  corresponding to the two Casimirs
of the Poisson structure $B$.

At the north and south poles of the $\bm{P}$ sphere the eigenvalues are $\lambda=\pm ap_{z}\pm i\beta$
where $\beta\in\R\backslash\left\{ 0\right\} $ is an arbitrary
parameter and $p_{z}=\pm\sqrt{2E}$. This implies
that the poles  of the $\bm{P}$-sphere are non-degenerate focus-focus points,
with  corresponding non-degenerate focus-focus value $(0,0)$.

At the equator of the $\bm{P}$ sphere the eigenvalues are $\lambda=0,0,\pm i\sqrt{m^{2}+\gamma^2}$.
Thus, all points on the equator of the $\bm{P}$ sphere are elliptic-transversal
critical points. 
\end{proof}

Note that for the elliptic-transversal points the vector field of $G$ is 
$\dot{\bm{P}}=2m \left(-p_{y},p_{x},0\right)$ and $\dot{\bm{L}}=\bm{0}$.
Thus for $m \not = 0$ the set of critical points in the preimage of $(m, m^2 - \gamma^2)$ 
is a periodic orbit along the equator of the $\bm{P}$-sphere. 
For $m=0$ this periodic orbit degenerates into a circle of fixed points, but from 
the point of view of the momentum map $(L_z, G)$ they are still non-degenerate.

In the general theory of semi-toric systems \cite{VuNgoc09} one simplifying assumption 
is that each focus-focus singular fibre only has a single focus-focus critical point in it.
Thus a doubly pinched torus as in the spheroidal harmonics system is not ``allowed''
in the original form of the theory. Even though it has recently been removed 
\cite{Pelayo19}  it is nevertheless interesting to note that the spheroidal harmonics system has a 
discrete symmetry that can be reduced such that the doubly pinched torus becomes
a reduced singly pinched torus.
In particular the two focus-focus critical points are identified with each other. 
When reducing by the full symmetry group the quotient
is not a smooth manifold. 
Consequently, we  quotient by $S_2$ only.

\begin{lem}
After discrete symmetry reduction by $S_2$, the reduced
phase space is $T^{*}\left(\mathbb{RP}^{2}\right)$ with fundamental
region chosen to be the northern hemisphere of the $\bm{P}$ sphere.
\end{lem}
\begin{proof}
The reduced $\bm{P}$ space is $\mathbb{RP}^{2}$ because $S_2$
identifies antipodal points of $S^2$. A fundamental region is the northern
hemisphere with $p_{z}\ge0$. Since $S_2$ does not act
on $\bm{L}$, the reduced phase space is therefore $T^{*}\left(\mathbb{RP}^{2}\right)$.
There is only one focus-focus point in the reduced phase space since
$S_2$ maps the north-pole and the south-pole of the sphere $S^2$ into each other.
\end{proof}

An interesting observation that can be made from Figure~\ref{fig:Glevel} (left) for $m=0$ is that because 
 $p = const$ on the critical level the action of the critical level is simply equal to the area of the rectangle 
 with side lentghs $2\gamma$ and $2$ divided by $2\pi$, so that $I_\eta(0,0) = 2\gamma/\pi$. 
 By Weyl's law this tells us that in the semiclassical limit (i.e.\ for large $\gamma$) the number of negative eigenvalues $g_l^0$ is 
 to leading order $2\gamma/\pi$, which for $\gamma = 16 $  gives approximately 10, which can be observed 
 in Figure~\ref{fig:Spheroidal-Eigenvalue}.

\appendix


\bibliographystyle{amsalpha}      
\bibliography{../../bib_cv/all,../../bib_cv/hd}{}   

\end{document}